\documentclass[aps,prx,twocolumn]{revtex4-2}

\usepackage{amsmath}
\usepackage{amssymb}
\usepackage{amsthm}
\usepackage{graphicx}
\usepackage{hyperref}
\usepackage{algpseudocode}
\usepackage{placeins}
\usepackage{xcolor}
\usepackage{CJKutf8}

\theoremstyle{plain}
\newtheorem{theorem}{Theorem}
\newtheorem{lemma}[theorem]{Lemma}
\newtheorem{corollary}[theorem]{Corollary}
\theoremstyle{definition}
\newtheorem{definition}[theorem]{Definition}
\theoremstyle{remark}
\newtheorem{remark}[theorem]{Remark}
\newcommand{\Opt}{\operatorname{Opt}}

\makeatletter
\def\skiplistcomma{%
  \@ifundefined{@listcomma}{}{%
    \global\let\saved@listcomma\@listcomma
    \gdef\@listcomma{\global\let\@listcomma\saved@listcomma}
  }%
}
\makeatother

\begin{document}

\title{Encoding Circuit Satisfiability in Rydberg Atom Arrays}

\author{Haotian Ji}
\thanks{These authors contributed equally to this work.}
\author{Zhangjie Qin,$^{*,}$\skiplistcomma}
\email[Contact author: ]{qinzhangjie@quantumsc.cn}
\author{Zheng An}
\author{Bowen Yan}
\author{Daoheng Niu}
\author{Kunzhe Dai}
\author{Jingkai Fang}
\author{Dongyang Cao}
\author{Jiangyu Cui}
\email[Contact author: ]{cuijy@quantumsc.cn}
\affiliation{Quantum Science Center of Guangdong-Hong Kong-Macao Greater Bay Area, Shenzhen 518045, China}

\begin{abstract}
   Rydberg atom arrays natively encode the maximum-weight independent set (MWIS) problem through the blockade
mechanism, so the Boolean circuit satisfiability problem (Circuit-SAT) can be
brought onto the platform once it is reduced to MWIS. The conventional
encoding of Circuit-SAT in the Rydberg atom array proceeds
through conjunctive normal
form (CNF) and incurs a
substantial atom overhead. We introduce
CAMERA (Circuit-SAT
Atom-efficient MWIS Encoding for Rydberg Arrays), a method that
provides MWIS encodings of Circuit-SAT instances on the
king subgraph geometry of the array. CAMERA represents each logic gate as
a compact weighted gadget and assembles the gadgets with a placement and
routing compiler inspired by very large scale integration (VLSI) design. On random multi-gate benchmarks, the
direct encoding route lowers the atom cost relative to the CNF
route by an average factor of $22.4 \pm 1.8$. To
demonstrate that the encoding extends from individual weighted gadgets to
multi-gate arithmetic blocks, we compile a full adder and a multiplier,
verifying each against its complete truth table by exact classical
ground state calculations. We further showcase solving a representative Circuit-SAT
instance end-to-end, from gate level compilation through a closed system
tensor-network simulation of a hardware compatible annealing protocol on the
encoded $30$-atom instance to readout of a satisfying assignment. These
results establish a complete encoding and simulation workflow as a proof
of principle, and a concrete route toward solving a broader family of
combinatorial problems on Rydberg atom arrays.
\end{abstract}

\maketitle

\section{Introduction}

The practical reach of quantum optimization will be set not only by the scale
and coherence of quantum processors, but also by how efficiently problems of
practical interest can be encoded in their native Hamiltonians. Before a
device can search for a low-energy state, the objective and constraints of the
original problem must be expressed using the interactions and controls
available on that platform. On analog hardware, this encoding is also a
physical construction. A reduction may be efficient in complexity-theoretic
terms and still be expensive in hardware: an auxiliary variable can require
an additional atom, a nonlocal constraint can require a chain of mediators,
and nonuniform weights can require site-resolved control. These costs are
explicit in Rydberg quantum-wire and arbitrary-connectivity
constructions~\cite{Kim2022RydbergWires,Nguyen2023ArbitraryConnectivity}, and
in experiments that implement weighted optimization through local light
shifts~\cite{deOliveira2025WeightedMWIS}. Application-oriented mappings and
studies of Rydberg encoding strategies further show that the chosen
representation strongly affects the size and structure of the instance that
can be implemented~\cite{Wurtz2024IndustryISet,Bombieri2025Encoding}. A
problem that is compact in its original description can therefore become much
larger after it is mapped to the device. Preserving the structure of the source
problem is consequently a central task of quantum encoding, rather than a
secondary software concern.

Rydberg atom arrays provide a particularly direct setting in which to address
this encoding problem. In the
strong-blockade regime, nearby atoms cannot be excited simultaneously, and
the low-energy configurations of the array encode independent sets of its
interaction
graph~\cite{Pichler2018QOMIS,Saffman2010RydbergQI,Browaeys2020RydbergReview}.
Site-dependent detunings extend this correspondence to the
MWIS problem by assigning an energy
bias to each vertex~\cite{deOliveira2025WeightedMWIS}. The platform has
progressed from one-dimensional chains of $51$ atoms to programmable
two-dimensional arrays and tweezer registers containing more than six
thousand atoms~\cite{Bernien2017ManyBody51,Ebadi2021Nature256,
	Scholl2021AFM2D,Bluvstein2022CoherentTransport,
	Bluvstein2024LogicalProcessor,Manetsch2025_6100Atoms}. Experiments have
realised maximum independent set (MIS) on unit-disk, Platonic, and king-subgraph (KSG)~\cite{Ebadi2022ScienceMIS,Byun2022Platonic,Kim2024MISDataset},
demonstrated quantum wires between separated
vertices~\cite{Kim2022RydbergWires}, and implemented weighted optimization
with local light shifts~\cite{deOliveira2025WeightedMWIS}. On selected hard
instances, Rydberg optimization has also shown superlinear scaling of the
exact-solution probability relative to a specified simulated-annealing
baseline~\cite{Ebadi2022ScienceMIS}. Application-oriented encodings,
local-detuning schemes, and hardness analyses have further developed this
hardware-native optimization model~\cite{Wurtz2024IndustryISet,
	Farouk2024MWIS,Yeo2024LocalDetuning,Cain2023FlatLandscapes,
	Andrist2023Hardness,Bombieri2025Encoding}. MWIS is therefore a useful target
model for encoding structured discrete problems to neutral atom hardware.

Circuit-SAT is a stringent test of such an encoder. Given a Boolean
circuit with a designated output, Circuit-SAT asks whether some input
assignment sets that output to one. It is
NP-complete~\cite{Cook1971TheoremProving,Karp1972Reducibility} and appears
throughout electronic-design automation, including combinational equivalence
checking and formal hardware
verification~\cite{Li2005FormalVerification}. A gate-level circuit is more
than a Boolean formula. It is a directed acyclic graph that records gate
boundaries, fan-out, and the reuse of intermediate signals. These features
make the representation compact and give a physical compiler useful
information about where logic should be placed and how signals should be
routed.

The standard route from Circuit-SAT to a Rydberg MWIS instance does not retain
this structure throughout the reduction. It first converts the circuit to an
equisatisfiable CNF formula through the Tseitin
transformation~\cite{Tseitin1983Complexity}, maps the clauses to an
independent-set graph, and embeds the graph's nonlocal edges with auxiliary
quantum wires~\cite{Kim2022RydbergWires,Nguyen2023ArbitraryConnectivity}.
The construction is polynomial and general, but it pays for this generality
before the geometry of the original circuit can be used. Tseitin variables,
literal vertices, consistency edges, and wire chains all become atoms in the
final layout. The physical size of the instance can therefore be dominated by
the intermediate representation rather than by the circuit itself.
Gadget-based encodings of $3$-SAT and graph coloring, automated gadget search,
and weighted copy-and-crossing constructions show that logical constraints
can instead be built from local modules~\cite{Jeong2023RydbergSAT,
	Angkhanawin2026GraphColoring,Pan2025TriangularGadgetSearch,
	Nguyen2023ArbitraryConnectivity}. What is missing is an automatic route from
a gate-level netlist to a blockade-valid atom layout that keeps the circuit
topology available throughout placement and routing.

Here we present CAMERA, a gate-level encoding of Boolean circuits into
weighted king-subgraph instances for Rydberg atom arrays. The overall
construction is shown in Fig.~\ref{fig:overview}. Each supported logic gate is
represented by a compact weighted gadget. Fan-out and crossing gadgets provide
the required wiring operations, and weighted atom chains carry signals between
separated ports. A feasibility-first, VLSI-inspired placement-and-routing
procedure assembles these components on the king grid and rejects layouts with
unintended blockade interactions. By bypassing the CNF representation, the
compiler lets physical design act directly on the original circuit topology.
For every validated layout, the MWIS manifold projects onto the input-output
relation of the circuit, so the ground states retain the logical meaning of
the source netlist.

The compiled graph also provides an exact readout of the Circuit-SAT answer. Branching the output port to one decides the instance. The optimal weight of the branched graph is compared with a reference weight, the optimal weight of the full graph. If the branched optimum, with the output port weight added back, equals the reference weight, the instance is satisfiable and the input port occupations spell a witness, whereas a strict deficit certifies unsatisfiability(Sec.~\ref{subsec:method-solving}). Branching the input ports instead
turns the same layout into a forward evaluator of the circuit
(Sec.~\ref{subsec:method-forward}).

We verify the construction from elementary gates to composed arithmetic
circuits. Exact ground-state calculations establish the complete truth-table
readout of a $30$-atom three-gate instance, an $85$-atom full adder, and a
$165$-atom two-bit multiplier. In the random multi-gate benchmark
reported here, the direct route uses roughly twenty-fold fewer atoms, a
reduction of about $95\%$, than the CNF-mediated baseline; for the full adder
and multiplier, the reduction is about thirty-fold. A closed-system
tensor-network simulation of the $30$-atom instance reaches the target MWIS
manifold in $58\%$ of projective samples under a $4\,\mu\mathrm{s}$
neutral-atom annealing schedule. Representative compiled layouts are further
tested numerically, linking the exact graph construction to physical state
preparation and readout. The calculations, compiler benchmarks, and
simulations together show that the gate-level representation survives the
full path from a Boolean netlist to a Rydberg array while reducing the physical
footprint by more than an order of magnitude.

The rest of the paper is organised as follows. Section~\ref{sec:background}
reviews Circuit-SAT, the Rydberg MWIS mapping, and the CNF-mediated baseline.
Section~\ref{sec:method} presents CAMERA, the branching
procedures, and the placement and routing compiler.
Section~\ref{sec:results} reports the resource comparison, the compiled
examples, and the dynamical and experimental results, and
Sec.~\ref{sec:conclusion} concludes with an outlook.

\section{Background}
\label{sec:background}
This section provides the background knowledge for the encoding
developed in this work and fixes the vocabulary used throughout the
paper. Section~\ref{subsec:csat-problem} defines the Circuit-SAT
problem. Section~\ref{subsec:rydberg-arrays} reviews the MWIS problem
and its native realization on the Rydberg atom array, and
Sec.~\ref{subsec:csat-to-mis} recalls the conventional CNF-mediated
reduction that serves as the baseline for comparison.

\begin{figure*}[t]
\centering
\includegraphics[width=\textwidth]{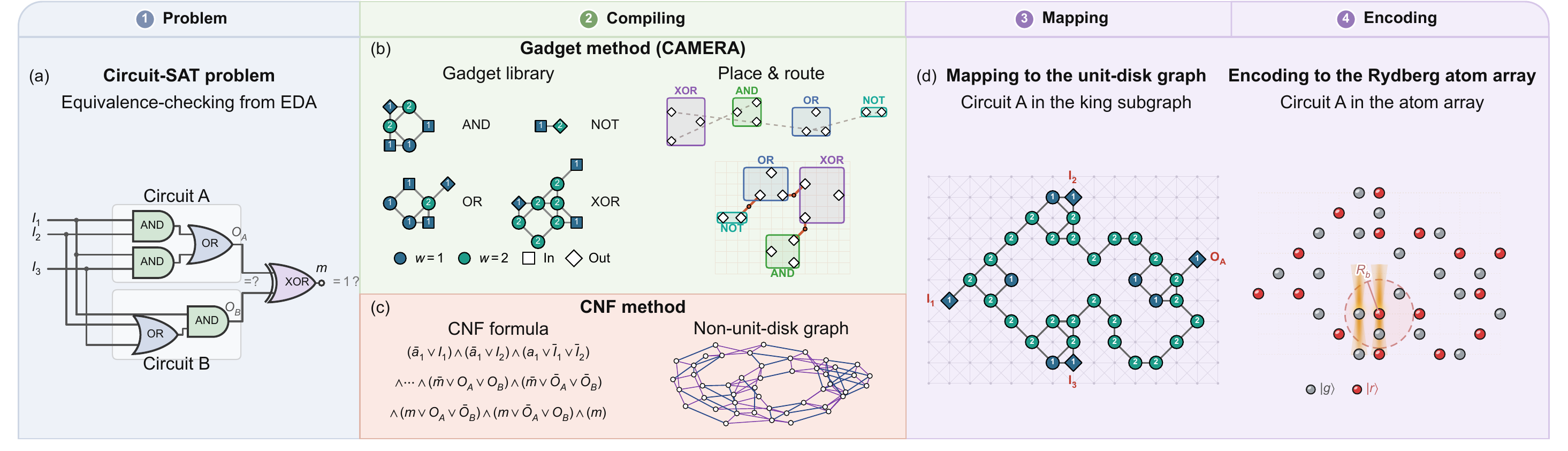}
\caption{Overview of the encoding framework, from an electronic-design
consistency check to a Rydberg native MWIS problem. The gadget-based route of this work (CAMERA) runs
(a)\,$\to$\,(b)\,$\to$\,(d); the conventional CNF route,
(a)\,$\to$\,(c).
\textbf{(a)}~The input: combinational equivalence checking by a
\emph{miter}. Two structurally different circuits $A$ and $B$ share the inputs $I_{1},I_{2},I_{3}$ and have their
outputs compared by an XOR gate, whose output $m=O_{A}\oplus O_{B}$
satisfies $m=1\Leftrightarrow A\neq B$.
\textbf{(b)}~The gadget method of this work. Left, the weighted KSG gadget library: one gadget per Boolean primitive, whose MWIS configurations
reproduce the gate truth table. Color encodes the on-site weight
$w\in\{1,2\}$, realized as a local detuning. Right, VLSI style
placement and routing of a generic four-gate netlist, from the unplaced
netlist with dashed airline connections (top) to the placed layout with
routed wire atoms (bottom).
\textbf{(c)}~The conventional route
bypassed by this work: the circuit is flattened into a Tseitin CNF formula (left) and mapped to an independent set problem
on a conflict graph (right). 
\textbf{(d)}~Mapping and encoding, shown for circuit $A$ of~(a). Left,
the compiled weighted KSG on the faint king lattice fabric. Right, its
Rydberg atom array realization: each atom sits in an optical tweezer, the
blockade disk of radius $R_{b}$ enforces the KSG edges, and red
$|r\rangle$ atoms mark one exact MWIS optimum.}
\label{fig:overview}
\end{figure*}

\subsection{The Circuit-SAT Problem}
\label{subsec:csat-problem}
A standard Boolean circuit $C$ is represented by a finite labeled directed
acyclic graph (DAG) $D_C=(V_C,A_C)$, whose vertices are inputs or gates and whose directed arcs are wires. The conventional De Morgan gate basis is
$\mathcal B_0=\{\land,\lor,\neg\}$, with binary AND and OR and unary NOT. This construction also admits binary XOR and therefore uses $\mathcal B=\mathcal B_0\cup\{\oplus\}$. For a given input assignment $x$, evaluating the gates of its DAG $D_C$ in topological order yields the output $F_C(x)$, thereby realizing the Boolean mapping $F_C:\{0,1\}^n\to\{0,1\}^m$ implemented by the circuit.

Circuit-SAT is an NP-complete decision problem with
important applications in electronic design automation (EDA). Given a Boolean
circuit $Q:\{0,1\}^{n_Q}\to\{0,1\}$, it asks whether there exists an input
$x\in\{0,1\}^{n_Q}$ such that $Q(x)=1$\cite{Cook1971TheoremProving,Karp1972Reducibility}. One important EDA
application of Circuit-SAT is combinational equivalence checking (CEC), which,
given an implementation circuit $C$ and a reference circuit $C_{\mathrm{ref}}$
with corresponding inputs and outputs, asks whether $F_C(x)=F_{C_{\mathrm{ref}}}(x)$ for every input $x\in\{0,1\}^n$.

\subsection{MWIS on Rydberg Atom Arrays}
\label{subsec:rydberg-arrays}
The second ingredient of the construction is the MWIS
problem. Given an undirected
graph $G=(V,E)$, distinct from
the circuit graph of Sec.~\ref{subsec:csat-problem}, with a positive weight
$w_{v}$ attached to each vertex $v\in V$, a subset
$S\subseteq V$ is \emph{independent} if no two of its
vertices are joined by an edge of $E$; the MWIS problem asks
for an independent set with largest total weight,
\begin{equation}
W^{\star}(G)\;=\;\max_{\substack{S\subseteq V\\ S\ \text{independent}}} W(S),
\qquad
W(S)\;=\;\sum_{v\in S} w_{v},
\label{eq:mwis}
\end{equation}
and we write $S^{\star}$ for an optimal set attaining $W^{\star}(G)$.

This problem can be solved natively on a Rydberg atom array, in which
each atom is held in an optical tweezer and driven
coherently between its ground state and a Rydberg state $|r\rangle$.
The driven array realizes an effective Ising type Hamiltonian
$H/{\hbar}=\sum_{i}\tfrac{\Omega_{i}}{2}\sigma_{i}^{x}-\sum_{i}\Delta_{i}n_{i}
+\sum_{i<j}V_{ij}n_{i}n_{j}$, where $\Omega_{i}$ is the Rabi frequency
applied to atom $i$, $\Delta_{i}$ is its local detuning, $n_{i}$ is its
Rydberg state occupation operator, and $\sigma_{i}^{x}$ is the Pauli-$x$
operator on its ground--Rydberg two-level
system~\cite{Jaksch2000FastGates,Saffman2010RydbergQI,Browaeys2020RydbergReview}.
The van der Waals interaction $V_{ij}=C_{6}/d_{ij}^{6}$ falls off
steeply with the interatomic distance $d_{ij}$, so for atoms closer than a
characteristic blockade radius $R_{b}$ the doubly excited configuration is
shifted far out of resonance and simultaneous excitation is suppressed. This
is the Rydberg blockade~\cite{Jaksch2000FastGates,Lukin2001DipoleBlockade}
[Fig.~\ref{fig:overview}(d)], which has been controlled with high fidelity
and scaled to programmable arrays of hundreds of
atoms~\cite{Bernien2017ManyBody51,Levine2018HighFidelity,Ebadi2021Nature256,Bluvstein2022CoherentTransport}.
Because two adjacent atoms cannot be simultaneously excited in this limit, every
blockade compatible excitation pattern corresponds to an independent set
of the blockade graph [Fig.~\ref{fig:overview}(d)]. Driving the system
into its low energy manifold at positive detuning favors configurations
with as many excitations as the blockade permits, so the ground state of
the array encodes a MIS of the programmed
geometry~\cite{Pichler2018QOMIS,Ebadi2022ScienceMIS,Nguyen2023ArbitraryConnectivity}.
The same correspondence extends directly to MWIS: choosing the local
detuning as $\Delta_{i}=\lambda w_{i}$, with $\lambda>0$, makes the energy
gain from exciting atom $i$ proportional to its vertex weight $w_{i}$. The
ground state occupation pattern then represents an MWIS $S^{\star}$ of
Eq.~\eqref{eq:mwis}, as demonstrated experimentally using local light
shifts~\cite{deOliveira2025WeightedMWIS}.

In the hard-blockade limit, a Rydberg atom array provides a direct
physical encoding of MIS graphs~\cite{Clark1991UnitDiskGraphs,Pichler2018QOMIS,Ebadi2022ScienceMIS}.
Specifically, the atomic arrangement realizes $G_{\mathrm{MIS}}$
geometrically: each atom
corresponds to a vertex $i\in V$, and an edge
$(i,j)\in E$ is present whenever the interatomic distance
$d_{ij}$ is smaller than the blockade radius $R_{b}$. A graph constructed
in this way---vertices placed in the plane, adjacent exactly when their
separation falls below a fixed radius---is a \emph{unit disk
graph}~\cite{Clark1991UnitDiskGraphs}. Its adjacency is inherited
entirely from the layout, so a graph carrying arbitrary long-range
connectivity admits no such realization. Finding an MIS is NP-hard on
general graphs~\cite{Karp1972Reducibility} and remains NP-hard on unit
disk graphs~\cite{Clark1991UnitDiskGraphs}, so this geometric
restriction does not diminish the expressive power of the target
problem. One example of a unit disk graph is the KSG,
induced when the atoms occupy sites of the integer grid: $V$ is a
finite subset of $\mathbb{Z}^{2}$, and two sites are adjacent exactly
when their Chebyshev distance is one,
\begin{equation}
E=\bigl\{\{u,v\}\subset V:\,
0<\lVert u-v\rVert_{\infty}\le 1\bigr\}.
\label{eq:ksg}
\end{equation}
This lattice form is the standard geometry of current tweezer arrays,
and it is adopted throughout this paper: Circuit-SAT is compiled onto
exactly this hardware native MWIS encoding.

\subsection{The CNF-Mediated Mapping from Circuit-SAT to Maximum
Independent Set}
\label{subsec:csat-to-mis}
The two ingredients above are conventionally connected by a three-stage
chain of reductions through CNF
[Fig.~\ref{fig:overview}(c)]. The circuit is first rewritten as an
equisatisfiable CNF formula via the Tseitin
encoding~\cite{Tseitin1983Complexity}, which introduces one auxiliary
variable per gate output and a fixed set of clauses enforcing each
gate's input--output consistency [Fig.~\ref{fig:overview}(c), left]; the
formula is then mapped to MIS by the classical clause-clique
construction, in which each clause becomes a clique of literal vertices
and an edge joins every literal to each occurrence of its negation, so
that the formula is satisfiable if and only if the graph admits an
independent set containing one vertex from every
clique~\cite{deOliveira2025WeightedMWIS}. Because the inter-clause
consistency edges give this conflict graph arbitrary connectivity, in
violation of the unit-disk restriction of the blockade
[Fig.~\ref{fig:overview}(c), right], a final embedding stage must
flatten it onto the plane, mediating every non-local edge by an
auxiliary quantum-wire atom
chain~\cite{Kim2022RydbergWires,Nguyen2023ArbitraryConnectivity}.
Tseitin variables, clause cliques, and wire chains each inflate the atom
count, so the compiled instance climbs into the thousands of atoms
already for circuits of a few gates (quantified in
Sec.~\ref{sec:resources}); this CNF-mediated chain is the baseline
against which our construction is measured, whereas the route developed
in this work [Fig.~\ref{fig:overview}(b)] encodes the circuit directly
into a weighted king subgraph and is the subject of the next section.

\section{Method}
\label{sec:method}
Our approach consists of four components, presented in turn below.
Section~\ref{subsec:method-branching} introduces the graph level
branching primitive that pins a single bit of a compiled instance.
Section~\ref{subsec:method-solving} builds on it to encode a Boolean
circuit into a weighted king subgraph, with the theorems and the output
branching procedure that decide and solve the instance.
Section~\ref{subsec:method-compilation} describes the placement and
routing compiler that produces the physical layouts, and
Sec.~\ref{subsec:method-forward} closes with the input branching
variant that evaluates the compiled circuit in the forward direction.

\subsection{The Branching Primitive}
\label{subsec:method-branching}
All readout procedures in this work rest on a single graph operation,
which we call \emph{branching}: forcing a designated vertex $v$ into,
or out of, the independent set being optimized in
Eq.~\eqref{eq:mwis}. We write
$N[v]$ for the closed neighborhood of $v$, the vertex together with its
neighbors, and $G-U$ for the graph induced after deleting a vertex
subset $U$.

Either constraint reduces to an unconstrained optimum on a smaller
graph. Forcing $v$ into the independent set excludes every neighbor of
$v$, so the largest weight of an independent set containing $v$ is
$w_{v}+W^{\star}(G-N[v])$. Forcing $v$ out merely removes it as a
candidate, so the largest weight of an independent set avoiding $v$ is
$W^{\star}(G-v)$. Since every independent set either
contains $v$ or does not, the two branches together recover the
unconstrained optimum,
\begin{equation}
W^{\star}(G) \;=\; \max\bigl\{\, w_{v} + W^{\star}(G - N[v]),\;
W^{\star}(G - v) \,\bigr\}.
\label{eq:branch}
\end{equation}
This case distinction is the classical branching step at the core of
exact algorithms for the maximum independent set
problem~\cite{TarjanTrojanowski1977,FominKratsch2010}. Here branching pins one bit: deleting $N[v]$,
with the weight of $v$ credited, asserts the bit to be $1$, and
deleting $v$ alone asserts it to be $0$. Applied to the output port,
this primitive decides and solves the Circuit-SAT instance
(Sec.~\ref{subsec:method-solving}); applied to the input ports, it
evaluates the compiled circuit on a specified input assignment
(Sec.~\ref{subsec:method-forward}).

\subsection{Solving Circuit-SAT via the MWIS Solution}
\label{subsec:method-solving}
The circuit is compiled into a weighted king subgraph
[Sec.~\ref{subsec:rydberg-arrays}, Eq.~\eqref{eq:ksg}] in two steps.
First, each Boolean primitive of the basis $\mathcal{B}$ is encoded by a
compact weighted KSG \emph{gadget}: a cluster of atoms with integer
weights $w\in\{1,2\}$, among which designated \emph{port} atoms carry the
logical signals, whose maximum weight independent sets reproduce exactly
the truth table of the gate [Fig.~\ref{fig:overview}(b)]. Second, the
gadgets are composed following classical Boolean wiring: when a gate
output feeds a downstream input the two port atoms are identified into a
single atom of summed weight, or are joined by a wire
chain~\cite{Nguyen2023ArbitraryConnectivity}. In either case the
composition is subject to the \emph{clean edge condition}: the composite
carries no edge beyond those internal to the gadgets and wire chains.
The result is a weighted graph $G_{C}$ with port atoms
$P_{C}=(I_{1},\ldots,I_{n},O_{1},\ldots,O_{m})$. Every independent set $S$
assigns a bit to each port, $1$ if the port atom belongs to $S$ and $0$
otherwise; these bits read across the ports form the \emph{port
projection} of $S$. We write $T(G_{C},P_{C})$ for the port projections
attained by the maximum weight independent sets of $G_{C}$,
$\mathcal{T}(C)=\{(x,F_{C}(x)):x\in\{0,1\}^{n}\}$ for the truth relation
of the circuit (Sec.~\ref{subsec:csat-problem}), and abbreviate
$W^{\star}\equiv W^{\star}(G_{C})$. The gadget library, obtained with the
automated search of Ref.~\cite{Pan2025TriangularGadgetSearch}, is
catalogued with its optimality certificates in the Supplemental Material
(Sec.~\ref{sm:gadget-design}); the placement and routing compiler that
produces the layouts is described in
Sec.~\ref{subsec:method-compilation}.

\begin{theorem}[Circuit realization]
\label{thm:realization}
For every Boolean circuit $C$ over the basis $\mathcal{B}$ compiled
as above, the port projections of the maximum weight independent sets
of $G_{C}$ are exactly the rows of the circuit truth table,
\begin{equation}
T(G_{C},P_{C})=\mathcal{T}(C).
\label{eq:realization}
\end{equation}
Moreover, the optimal weight $W^{\star}$ is the sum of the optimal
weights of the constituent gadgets and wire chains.
\end{theorem}

\noindent Theorem~\ref{thm:realization} states the soundness and
completeness of the encoding: every optimum is \emph{gate consistent},
its port occupations spelling a valid input assignment together with
the value the circuit computes on it, and conversely every truth
table row is realized by at least one optimum. The proof is
given in the Supplemental Material (Sec.~\ref{sm:composition-theory}).

Deciding satisfiability now reduces to one application of the
branching primitive. Let $v_{O}$ be the atom of the designated output
port, of weight
$w_{O}$, and let
$W_{\mathrm{br}}=W^{\star}(G_{C}-N[v_{O}])$ be the branched optimum
appearing in the first argument of Eq.~\eqref{eq:branch}.

\begin{theorem}[Output branch SAT criterion]
\label{thm:branch-criterion}
The first branch of Eq.~\eqref{eq:branch} bounds the branched optimum by
$W^{\star}-W_{\mathrm{br}}\ge w_{O}$, and the instance is satisfiable if
and only if this bound is saturated,
\begin{equation}
W^{\star}-W_{\mathrm{br}}\;=\;w_{O}.
\label{eq:branch-criterion}
\end{equation}
A strict deficit $W^{\star}-W_{\mathrm{br}}>w_{O}$ is therefore the only
alternative, and certifies that no gate consistent configuration attains
$O=1$.
\end{theorem}

\noindent The criterion follows by combining this bound with
Theorem~\ref{thm:realization}
(Supplemental Material, Sec.~\ref{sm:output-branch}).

The reference weight $W^{\star}$ requires no optimization. By the
additivity clause of Theorem~\ref{thm:realization} it is the sum of the
optimal weights of the gadgets and wire chains that the compiler places,
all of which are fixed once placement and routing terminate
(Sec.~\ref{subsec:method-compilation}). The decision therefore rests on a
single MWIS computation. The satisfying condition $O=1$ is asserted by
the branching primitive of Sec.~\ref{subsec:method-branching}, the MWIS
problem is solved on the branched graph alone, and the resulting
$W_{\mathrm{br}}$ is compared with $W^{\star}$ through
Eq.~\eqref{eq:branch-criterion}. A satisfiable instance returns its
witness from the same computation, since the input port occupations of
the branched optimum spell a satisfying assignment; an unsatisfiable one
is certified by the strict deficit. Solving the unconditioned graph
remains available as a separate readout---with all ports free its optima
range over the complete truth table---but the decision does not require
it. The procedure is summarized in Algorithm~S1 of the Supplemental
Material (Sec.~\ref{sm:output-branch}) and demonstrated on a minimal
unsatisfiable instance in Sec.~\ref{subsec:unsat}
(Fig.~\ref{fig:unsat}).

\subsection{Compilation Method}
\label{subsec:method-compilation}
The clean edge condition of Sec.\ref{def:pin-amalgamation}, is also the central compilation constraint: every intended adjacency must be placed within the king distance of Eq.~\eqref{eq:ksg}, and every other pair of atoms farther apart. A general circuit also has nonlocal connections that cannot all be realized by overlapping pins; these are carried by wire chains and, in the plane, by explicit crossing gadgets. The compiler therefore optimizes the coarse topology and crossings first, then expands and refines the atom level layout, and finally assigns weights to emit the weighted graph $(G_{C},w)$.

\emph{Abstract stage.} The abstract stage fixes connectivity and
relative topology before atoms are placed. Each logic or fanout gadget
is assigned an \emph{anchor} on a coarse routing grid, and each pin is
an offset from its anchor with an allowed escape direction. Primary
inputs and outputs are contracted onto their consuming or producing
pins rather than placed as movable gadgets. For a fixed placement,
every source--destination pin pair is routed by A$^{\star}$ on the
eight-neighbor grid~\cite{Hart1968AStar} under the clean edge
constraint. A route leaves and enters pins along their escape
directions and makes no unintended contact; the sole exception is a
legal crossing, which is materialized as a crossing gadget. Because
sequential routing is order dependent, conflicts are resolved by
PathFinder-style history-based rerouting~\cite{McMurchie1995PathFinder}.
The layout is then optimized in two steps (Algorithm S2): an ensemble of position-only simulated-annealing runs\cite{Kirkpatrick1983Annealing,Metropolis1953EquationState} first collects complete-valid layouts, and greedy descent which compacts each collected layout.

\emph{Real stage.} The real stage turns each archived abstract layout
into an atom level layout (Algorithm~S3). Strictly increasing row and
column maps expand the coarse grid until every rigid gadget template
fits at its mapped anchor, and the templates are stamped with their
abstract orientations. Each split net is reconnected by A$^{\star}$
anchored to its abstract route, with consecutive route points realizing
intended KSG edges and all other pairs held outside blockade distance.
To realize the wire relation of the composition theory, every route
must contain an odd number of points; even-parity routes are repaired
by a local splice or by parity-aware rerouting. A validated realization
is then descended greedily to reduce the atom count .

\emph{Weight assignment.} Finally, gate, fanout, and crossing atoms
retain their template weights, interior wire atoms take weight $2$, and
each identified endpoint adds the wire's unit contribution to the
original pin weight (Sec.~\ref{sm:composition-theory}). Since the library
weights are $1$ and $2$ and every identified pin joins a weight-$1$
gadget port to a weight-$1$ chain endpoint, no rule raises an atom above
weight $2$: the binary weight alphabet $w\in\{1,2\}$ is therefore
preserved under composition, and the compiled array requires only two
local-detuning values however large the circuit. Odd wire parity and
clean pin identification realize exactly the hypotheses of the
composition theory: the amalgamation of
Definition~\ref{def:pin-amalgamation} and the wire insertion of
Corollary~\ref{cor:wire-insertion}. The Circuit-SAT to MWIS
equivalence of the emitted $(G_{C},w)$ therefore follows from
Theorem~\ref{thm:realization} by construction, with no enumeration of
the final MWIS manifold. The two stages are given in pseudocode as
Algorithms~S2 and~S3 of the Supplemental Material
(Sec.~\ref{sm:compilation-details}).

\subsection{Forward Circuit Evaluation via the MWIS Solution}
\label{subsec:method-forward}
Solving Circuit-SAT is intrinsically a \emph{backward} problem: the
output is asserted and the branched optimization propagates that
constraint against the direction of computation. The same compiled graph
also evaluates the circuit \emph{forward}, by applying the branching
primitive of Sec.~\ref{subsec:method-branching} to the input ports
rather than to the output. Given an input assignment
$x=(x_{1},\ldots,x_{n})$, each input port atom is branched to its bit
value and the MWIS problem is solved on the resulting reduced graph. By
Theorem~\ref{thm:realization} every optimum of that graph is gate
consistent and compatible with the pinned inputs, so the occupations of
the output port atoms read off $F_{C}(x)$. 
The forward reading is verified exhaustively for the compiled three-gate
circuit in the Supplemental Material, where each of the eight input
assignments reproduces its truth table row and any incorrect output
costs exactly one weight unit (Sec.~\ref{sm:cir1-verification},
Table~\ref{tab:sm-cir1-truth}).

Nothing in the branching primitive distinguishes inputs from outputs.
Equation~\eqref{eq:branch} applies to any port atom, so an arbitrary
subset $Q\subseteq P_{C}$ may be pinned to prescribed bits, asserting $1$
by deleting $N[v]$ and crediting $w_{v}$ and asserting $0$ by deleting
$v$ alone. Writing $W_{Q}$ for the optimum of the reduced graph together
with the credited weights, one has $W_{Q}\le W^{\star}$, with equality if
and only if some row of $\mathcal{T}(C)$ agrees with the prescribed bits;
the optima attaining it project onto exactly those rows. Pinning a mixed
subset of ports therefore answers a general constraint query on the
circuit, of which the satisfiability decision of
Sec.~\ref{subsec:method-solving}, with $Q$ the output port pinned to $1$,
and the forward evaluation above, with $Q$ the full set of input ports,
are the two extreme cases.

\begin{figure}[!tp]
\centering
\includegraphics[width=\columnwidth]{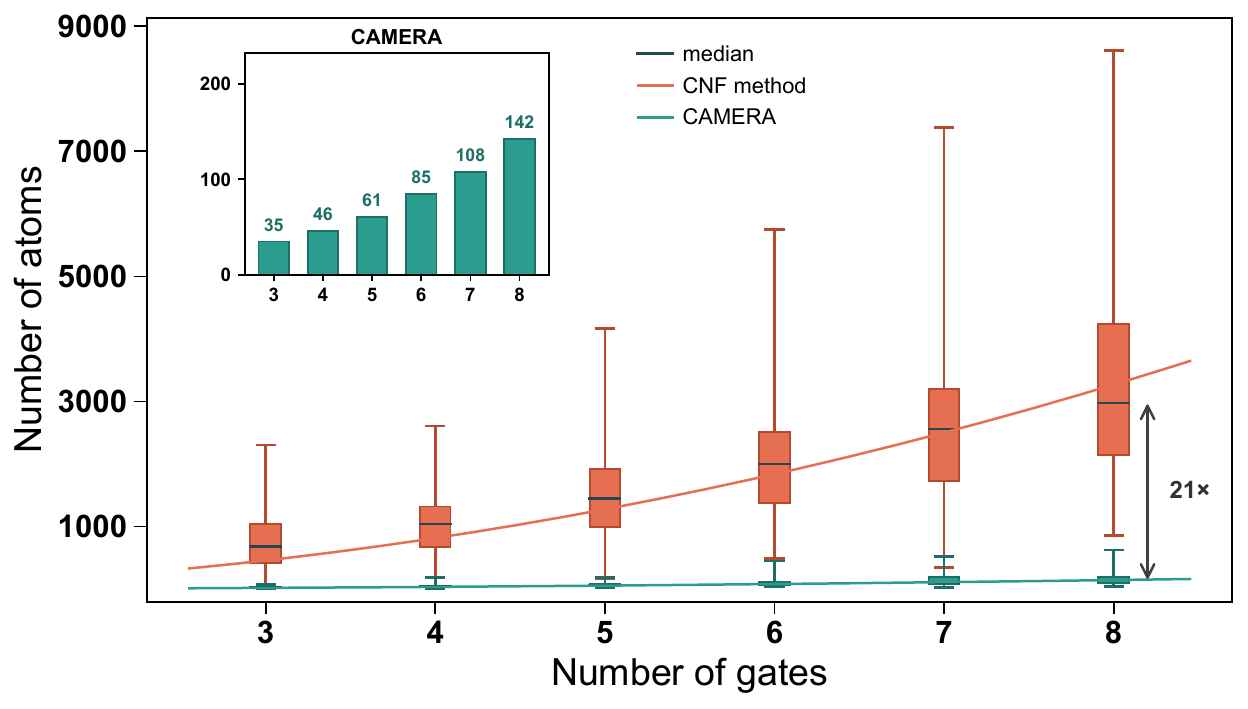}
\caption{Atom cost of the conventional CNF route (red) and the direct
gadget route (CAMERA, teal) on random benchmark circuits
of $g = 3$--$8$ gates ($100$ random circuits per gate count on each
route). For each gate count, the
thin vertical line spans the full min--max range, the box the
interquartile range, and the dark slate tick marks the median. The
curves are pure quadratics $\mathrm{atoms} = n\,g^{2}$ fitted through
the medians, and the double-headed arrow marks the $21\times$ median
suppression at $g = 8$. Inset: atom count of the median CAMERA instance
at each gate count.}
\label{fig:two-path}
\end{figure}

\section{Main Results}
\label{sec:results}

Section~\ref{sec:resources} opens with the headline payoff of the
approach, the atom count advantage of the direct gate level route
(CAMERA) over the conventional CNF-based reduction. The remaining
subsections build the constructive evidence behind this comparison,
from elementary correctness to multi gate scale.
Section~\ref{subsec:gate-mapping} validates the gate gadget library.
Section~\ref{subsec:small-circuit} demonstrates the complete workflow
on a hand-checkable three-gate instance, and
Sec.~\ref{subsec:annealing} simulates its annealing dynamics.
Section~\ref{subsec:unsat} certifies an unsatisfiable instance, and
Sec.~\ref{subsec:arith-blocks} verifies the compilations of two
arithmetic blocks.

\subsection{Resource Comparison}
\label{sec:resources}

The practical payoff of encoding at the level of logic gates is a
substantial reduction in atom count relative to the conventional CNF
route (Method~1 of Sec.~\ref{subsec:csat-to-mis}), whose Tseitin
variables, clause-clique padding, and quantum wire
chains~\cite{Tseitin1983Complexity,Kim2022RydbergWires,Nguyen2023ArbitraryConnectivity}
make a large fraction of the final array encode reduction structure
rather than circuit logic. The direct route eliminates the CNF
intermediate, turning each gate into a compact weighted KSG gadget whose
geometric flexibility the placement and routing compiler exploits.
Figure~\ref{fig:two-path} quantifies the gap on random benchmark
circuits of three to eight gates, $100$ circuits per gate count on each
route. The CNF medians follow the pure quadratic mapping $51.1\,g^{2}$
and climb to nearly $3000$ atoms by eight gates, whereas the
gadget
medians grow monotonically from $35$ to only $142$ atoms over the same
window. Forcing the same quadratic form through them yields a
coefficient of only $2.29$, and although this fit is as tight as the CNF
one ($R^{2} = 0.96$ versus $0.94$), over so narrow a window the gadget
medians remain nearly as
consistent with linear growth. The median CNF-to-gadget ratio ranges
from $19.5$ to $23.8$ across the window, averaging $22.4 \pm 1.8$
(one standard deviation across the six gate counts),
and the arrow
at eight gates marks the $21\times$ suppression there---a structural,
more than twenty-fold saving that brings practically sized Circuit-SAT
instances
within reach of current neutral atom hardware.

\begin{figure*}[!tp]
\centering
\includegraphics[width=\textwidth]{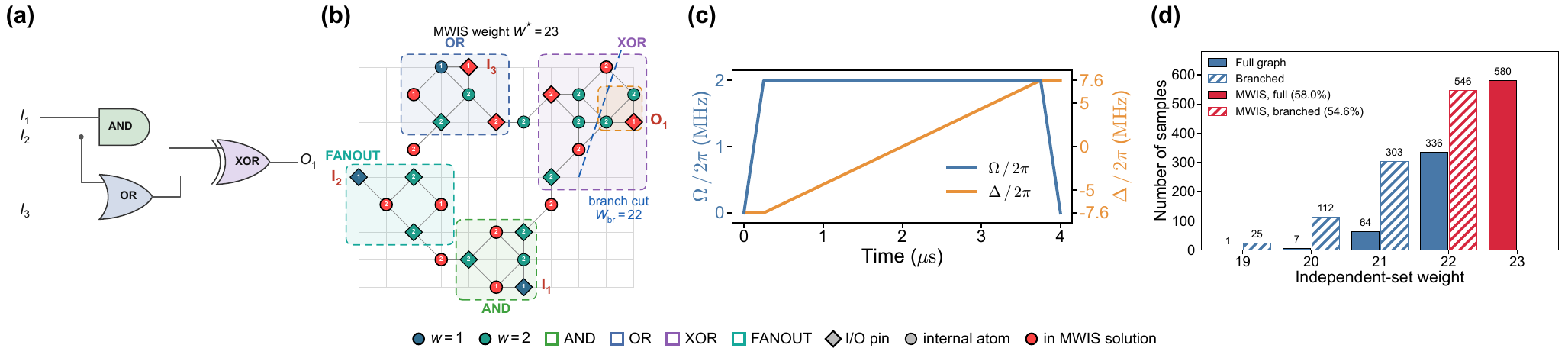}
\caption{End-to-end compilation and solution of the multi-gate Circuit-SAT
instance $C(I_{1}, I_{2}, I_{3}) = (I_{1}\wedge I_{2}) \oplus (I_{2}\vee I_{3})$.
\textbf{(a)}~A three-gate circuit combining AND, OR, and XOR, with output
$O_{1}$.
\textbf{(b)}~The compiled $30$-atom weighted king subgraph: colored
dashed boxes outline the four gadgets (AND, OR, XOR, FANOUT), node color
encodes the on-site weight $w\in\{1,2\}$, and diamonds mark the
pin atoms, with the ports $I_{1},I_{2},I_{3},O_{1}$ labeled in red. Red
fill marks one of the eight degenerate MWIS optima ($W^{\star}=23$),
encoding the input row $(I_{1},I_{2},I_{3})=(0,0,1)$ with $O_{1}=1$. The
orange frame and blue dashed cut mark the output branching that asserts
$O_{1}=1$: deleting the three framed atoms leaves the $27$-atom branched
graph with $W_{\mathrm{br}}=22$, and
$W^{\star}-W_{\mathrm{br}}=1=w_{O}$ meets the branch criterion of
Eq.~\eqref{eq:branch-criterion}, certifying satisfiability (Supplemental
Material, Sec.~\ref{sm:cir1-branched}).
\textbf{(c)}~Analog annealing pulse: the global Rabi amplitude
$\Omega(t)$ is ramped on and off while the site detunings $\Delta(t)$,
scaled by the integer weights, are swept linearly ($\Delta(t)$ shown
for the maximum-weight sites).
\textbf{(d)}~Independent set weight distributions over $1000$ projective
samples of the simulated final state, for the full graph (solid bars,
$988$ blockade consistent samples) and the branched graph (hatched bars,
$986$): each histogram peaks at the maximum weight of its graph (red
bars), $W^{\star}=23$ in $58.0\%$ and $W_{\mathrm{br}}=22$ in $54.6\%$
of the shots.}
\label{fig:small-circuit}
\end{figure*}

These scalings have a structural origin, separating the overhead
intrinsic to the problem from that introduced by the mapping. Any
reduction terminating on a two-dimensional unit disk geometry must
realize non-planar connectivity through explicit crossing structures: in
the crossing lattice construction underlying the CNF
route~\cite{Nguyen2023ArbitraryConnectivity}, each of the $N$ vertices
of the intermediate MIS graph becomes a copy line whose length is set by
the crossings it must traverse, so the embedded atom count scales as
$N \cdot pw$, with $pw$ the pathwidth of the vertex layout. Because
Circuit-SAT is NP-complete, no polynomial time reduction can bound the
connectivity of a general instance; a worst case circuit with
unstructured wiring forces $pw = \Theta(N)$ and hence a quadratic atom
cost $\Theta(N^{2})$---a ceiling that applies to \emph{both} routes,
since the gadget route must likewise route every wire crossing through
an explicit crossing gadget. The atom cost nevertheless remains
polynomial even for logically hard instances: the exponential price of
NP-hardness is paid in the ground state search, not in the size of the
array.

here the routes part ways is in the cost of a \emph{typical} circuit
and in the coefficient of the unavoidable quadratic. The CNF route inflates the vertex count before embedding, adding
roughly four to twelve literal-occurrence vertices per gate, and ties
every clause sharing a variable to the same few circuit inputs. Its
pathwidth therefore grows in proportion to $N$ already for ordinary
circuits, and the quadratic regime is generic, as the fitted
$51.1\,g^{2}$ expresses. The gadget
route instead spends a fixed handful of atoms per gate and pays routing
overhead only for the wiring the circuit actually contains. Ordinary
feed-forward circuits are dominated by short-range connections with few
crossings, so the typical atom cost grows nearly linearly over the
sampled window, and even the quadratic term that dense wiring would
activate carries a more than twenty times smaller coefficient ($2.29$ versus
$51.1$), because crossings are routed between compact gate gadgets
rather than inflated clause blocks. Both routes thus share the same
NP-hardness-imposed quadratic worst case, but the gadget route reduces
its prefactor by more than an order of magnitude and, for the structured
circuits of practical interest, replaces the quadratic behavior with
near-linear growth.

\subsection{Circuit Gate to KSG Mapping}
\label{subsec:gate-mapping}
The first validation layer concerns the translation of elementary circuit gates into
king-subgraph (KSG) gadgets. Each logical gate is represented by a specially designed
KSG gadget, and the four primitives AND, OR, XOR, and NOT constitute the single gate
library from which the compilation pipeline draws every gate it places, displayed in
the framework overview [Fig.~\ref{fig:overview}(b), left]. The design goal is that the
maximum-independent-set configurations of each gadget reproduce exactly the truth
table of the original Boolean gate, so that composing gadgets at the graph level
preserves the logical behavior of the circuit. The two-atom NOT gadget, whose single
king edge already enforces the required anti-correlation, is detailed together with
the other gadgets' full MWIS solution sets in the Supplemental Material
(Sec.~\ref{sm:not-gadget}).

For a given circuit primitive, the gadget search can be carried out by brute force
over small king-subgraph constructions, following the automated gadget-search strategy
introduced for encoding hard problems on Rydberg
lattices~\cite{Pan2025TriangularGadgetSearch} and adapted here to the king-subgraph
geometry: candidate gadgets with small integer vertex weights, restricted in our
library to the range $1$--$2$, are enumerated and tested for whether their optimal
independent sets realize the required gate behavior. This gives an explicit and auditable route to compact gate
gadgets rather than relying on heuristic design intuition. The complete
search formulation is detailed in the Supplemental Material
(Sec.~\ref{sm:search-formulation}). It includes the integer program
that fixes the vertex weights and the exhaustive polyplet enumeration
underlying the minimality certificates. The validation task is therefore twofold: verify
that each gadget reproduces the correct gate truth table, and confirm that it is
minimal or near-minimal within the explored search class. Across the resulting
four-gate library all ports carry weight $1$, so the analog implementation requires
only two local-detuning values and the ports look identical across all gadgets; this
gate-by-gate validation forms the most basic correctness check for the full
compilation pipeline.

\subsection{Circuit-SAT Demonstration: a Multigate Circuit Example}
\label{subsec:small-circuit}

Having validated the gate library at the level of individual primitives, we
now demonstrate the complete workflow of
Sec.~\ref{sec:method}---gadget composition, satisfiability decision, and
solution extraction---on a Circuit-SAT instance whose internal structure is
small enough to inspect by hand. The three-gate circuit of
Fig.~\ref{fig:small-circuit}(a) exercises all three two-input library
primitives: on the inputs $I_{1}, I_{2}, I_{3}$ it produces the single
output $O_{1}$, realizing the Boolean function
$C(I_{1}, I_{2}, I_{3}) = (I_{1}\wedge I_{2}) \oplus (I_{2}\vee I_{3})$.
Following the definition of Sec.~\ref{subsec:csat-problem}, the instance
asks whether there exists an input $x \in \{0,1\}^{3}$ with $C(x) = 1$,
together with the search problem of exhibiting a witness; both are settled
below by the output branching procedure of
Sec.~\ref{subsec:method-solving}.

The composition follows the wiring rules of
Sec.~\ref{subsec:method-solving}. The input $I_{2}$ is the only fanned out
signal: it drives both the AND gate and the OR gate, so a single fanout
gadget duplicates it, exposing two boundary atoms that carry the same
Boolean value in every MWIS configuration; the remaining gate-to-gate
links---the AND and OR outputs feeding the XOR gate---are realized by
shared port wiring, mediated by wire atoms where the placement keeps the
ports apart. No output-assignment preprocessing is applied: the output port
$O_{1}$ is retained as an explicit boundary atom, so the compiled graph
encodes the \emph{complete} truth table of $C$ rather than only its
satisfying rows. The result is the $30$-atom, $41$-edge weighted king
subgraph of Fig.~\ref{fig:small-circuit}(b), assembled from the $6$-atom
AND, $6$-atom OR, and $9$-atom XOR gadgets, the $5$-atom fanout gadget, and
four wire atoms. Each gadget keeps its verified template weights; the wire
atoms and the gadget ports engaged in inter-gadget connections are raised
from weight $1$ to $2$ by the wire rule of Sec.~\ref{sm:composition-theory}, which
leaves the optimal assignments unchanged, while the four primary ports
retain weight $1$, so the assembled graph keeps the binary weight
alphabet (Sec.~\ref{subsec:method-compilation}). For comparison, pushing the identical
three-gate circuit through the CNF chain of Sec.~\ref{subsec:csat-to-mis}
yields an intermediate MIS graph of $30$ vertices, already as many as the
entire gadget compilation, which the wire chain embedding then inflates to
an unweighted king subgraph of $1257$ atoms, roughly $42$ times the $30$
atoms compiled here.

The compiler fixes the reference weight of
Sec.~\ref{subsec:method-solving} without optimization,
$W^{\star} = 23$, and exact enumeration confirms both this value and the
encoding: the optimum is attained by exactly eight degenerate
independent sets, one per input assignment, whose port occupations
reproduce the circuit truth table row by row (Supplemental Material,
Sec.~\ref{sm:cir1-verification}). The eight optima contain $13$ or $14$
atoms while tying at the same weight, illustrating that the ground state
degeneracy is organized by weight rather than by cardinality; one
representative optimum is highlighted in Fig.~\ref{fig:small-circuit}(b),
the one encoding the input row $(I_{1},I_{2},I_{3})=(0,0,1)$, whose
occupied weight-$1$ output port reads $O_{1}=1$.
Deciding this Circuit-SAT instance then proceeds by
the output branching shown in the same panel: cutting the weight-$1$
output port and the two atoms it blockades leaves a $27$-atom branched
graph whose optimum---the highlighted set stripped of the $O_{1}$
atom---has weight $W_{\mathrm{br}} = 22$, so the branch criterion of
Eq.~\eqref{eq:branch-criterion} is met and the instance is certified
satisfiable. The witness comes from the same computation: the input port
occupations of the branched optimum spell out
$(I_{1},I_{2},I_{3}) = (0,0,1)$, for which indeed
$C(x) = (0 \wedge 0) \oplus (0 \vee 1) = 1$ (Supplemental Material,
Sec.~\ref{sm:cir1-branched}). This worked example is the elementary unit
of the broader compilation pipeline: the larger arithmetic blocks of
Sec.~\ref{subsec:arith-blocks} are obtained by iterating exactly the same
composition rules at scale. The analog quantum annealing dynamics of this
compiled instance [panels (c) and (d) of Fig.~\ref{fig:small-circuit}] are
analyzed in the following subsection.

\begin{figure*}[!tp]
\centering
\includegraphics[width=0.8\textwidth]{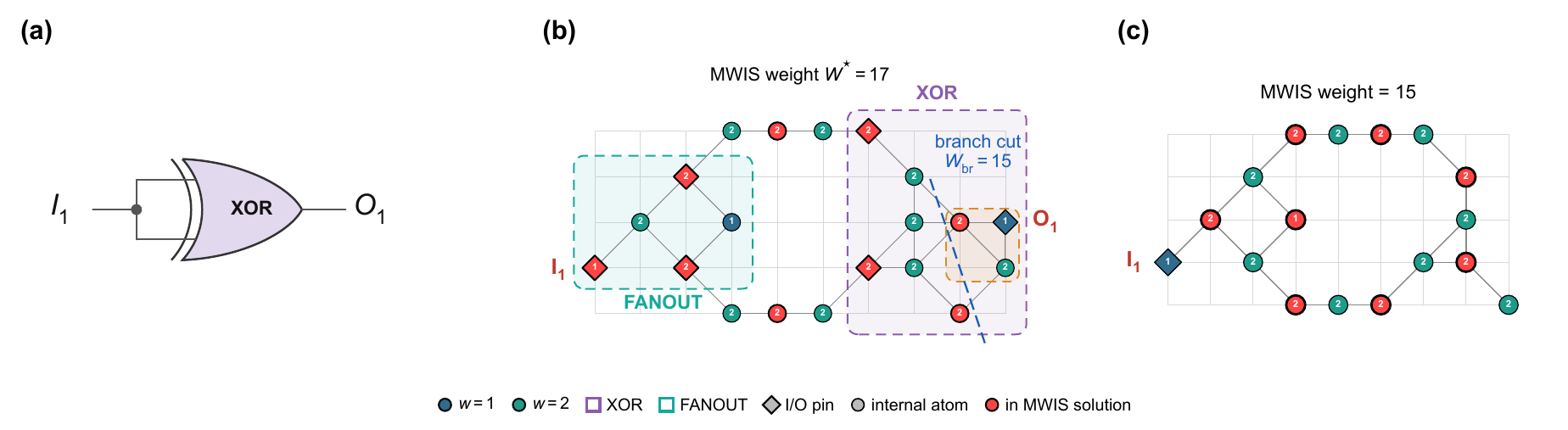}
\caption{Detecting unsatisfiability from the compiled king subgraph, for
the minimal unsatisfiable instance $O_{1} = I_{1} \oplus I_{1}$.
\textbf{(a)}~Logical circuit: the single input $I_{1}$ is fanned out into
both ports of an XOR gate, so $O_{1}$ is identically $0$ and the demand
$O_{1} = 1$ is unsatisfiable.
\textbf{(b)}~The compiled $20$-atom weighted KSG: dashed boxes outline the
FANOUT and XOR gadgets, node color encodes the on-site weight
$w\in\{1,2\}$, diamonds mark the pin atoms, and the ports
$I_{1}, O_{1}$ are labeled in red. Red fill marks one of the two
degenerate MWIS optima ($W^{\star} = 17$); every optimum reads
$O_{1} = 0$. The orange frame and blue dashed cut mark the output
branching that asserts $O_{1} = 1$: deleting the three framed atoms
leaves the branched graph of~(c).
\textbf{(c)}~The $17$-atom branched graph, with MWIS weight
$W_{\mathrm{br}} = 15$: since
$W_{\mathrm{br}} + w_{O} = 16 < 17 = W^{\star}$, the weight deficit
certifies unsatisfiability from a single optimization on the reduced
graph.}
\label{fig:unsat}
\end{figure*}

\subsection{Quantum Annealing Dynamics Simulation}
\label{subsec:annealing}
The practical relevance of the encoding depends not only on its static
correctness but also on whether realistic analog dynamics actually reach the
encoded optima. As an intermediate validation layer between graph level
verification and eventual hardware execution, we therefore simulate a
quantum annealing type protocol on the Rydberg Hamiltonian with tensor
network methods, adopting a standard neutral atom hardware model with
parameters in the range of existing Rydberg array platforms such as QuEra's
Aquila system~\cite{Wurtz2023Aquila}. The aim is not to claim a general
quantum advantage from the present data, but to verify that the compiled
Circuit-SAT-derived KSG instances display the expected solution-seeking
behavior under a physically motivated, hardware compatible schedule.

We run the protocol on the compiled multi-gate circuit instance of
Fig.~\ref{fig:small-circuit}(b), whose eight degenerate maximum weight
independent sets at $W^{\star}=23$ fix the target solution manifold before
any time evolution. The annealing pulse [Fig.~\ref{fig:small-circuit}(c)]
ramps the global Rabi drive $\Omega(t)/2\pi$ from zero to $2\,\mathrm{MHz}$
over the first $0.25\,\mu\mathrm{s}$, holds it flat, and ramps it back to
zero in the final $0.25\,\mu\mathrm{s}$, while the detuning $\Delta(t)/2\pi$
at the highest weight site is swept linearly from $-7.6\,\mathrm{MHz}$ to
$+7.6\,\mathrm{MHz}$ over the $3.5\,\mu\mathrm{s}$ plateau, with the
detunings on the remaining atoms scaled by their integer weights; the
instantaneous ground state thereby evolves from the all-empty state into the
MWIS-occupied configurations of the encoded graph. Of $1000$ projective
samples of the simulated final state, $988$ ($98.8\%$) satisfy every
pairwise blockade constraint, and their weight distribution
[Fig.~\ref{fig:small-circuit}(d), solid bars] is concentrated at the
optimum: $58.0\%$ of the shots attain $W^{\star}=23$, terminating exactly
in the eight-fold-degenerate ground state manifold that enumerates the
circuit truth table, with rapid decay at lower weights. The identical
protocol on the output branched $27$-atom graph [the branch cut of
Fig.~\ref{fig:small-circuit}(b)] performs equally well
[Fig.~\ref{fig:small-circuit}(d), hatched bars]: $986$ samples ($98.6\%$)
are valid and $54.6\%$ reach $W_{\mathrm{br}}=22$, the four-fold-degenerate
manifold of satisfying input rows, so the branched certificate instance is
sampled as efficiently as the instance it certifies. The concentration of
both distributions at the encoded optimum after a single
$4\,\mu\mathrm{s}$ sweep shows that the analog dynamics flow
preferentially towards the solution manifold---higher success
probabilities being attainable through longer annealing times or repeated
sampling---providing the dynamical validation layer anticipated above.

\subsection{Unsatisfiable Instances and Solution Self-Consistency}
\label{subsec:unsat}

The results presented so far concern satisfiable circuits, but the
negative answer carries as much operational weight as the positive one: in
the equivalence checking application of Sec.~\ref{subsec:csat-problem}, it
is precisely the \emph{unsatisfiability} of the miter that certifies two
designs equivalent. Since an optimizer always returns \emph{some}
configuration, we probe how the encoding behaves when no solution exists
on a minimal example, the single-gate instance
$O_{1} = I_{1} \oplus I_{1}$ [Fig.~\ref{fig:unsat}(a)]: because
$x \oplus x = 0$ for every Boolean $x$, the demand $O_{1} = 1$ is
unsatisfiable by construction. The same pipeline compiles it into the
$20$-atom weighted king subgraph of Fig.~\ref{fig:unsat}(b), small enough
that its full optimal set can be enumerated exactly (Supplemental
Material, Sec.~\ref{sm:unsat-enumeration}).

The enumeration shows how unsatisfiability manifests in the optimum. Both
degenerate maximum weight independent sets of the full graph, at the
compiler-fixed reference weight $W^{\star} = 17$
[Fig.~\ref{fig:unsat}(b)], are gate consistent and read
out $O_{1} = 0$, the correct value of $I_{1} \oplus I_{1}$; no optimum
places the output port at $1$. Certifying unsatisfiability this way,
however, would require enumerating the entire optimal set. Output
branching provides the economical test: deleting the $O_{1}$ port and the
atoms it blockades leaves the $17$-atom branched graph of
Fig.~\ref{fig:unsat}(c), whose single MWIS computation returns
$W_{\mathrm{br}} = 15$, and
$W_{\mathrm{br}} + 1 = 16 < 17 = W^{\star}$, so forcing the output to $1$
costs weight and the instance is certified unsatisfiable from a single
optimization on a smaller graph. This weight deficit is the quantitative
imprint of the inconsistent assumption $O_{1} = 1$; on analog hardware,
where individual samples may also fall below the optimum for purely
dynamical reasons, the same comparison applies sample by sample, and
persistent failure to reach the satisfiable reference weight is the
operational signature of an unsatisfiable instance.

\subsection{Arithmetic Blocks into KSG}
\label{subsec:arith-blocks}
\begin{figure*}[!t]
\centering
\includegraphics[width=0.68\textwidth]{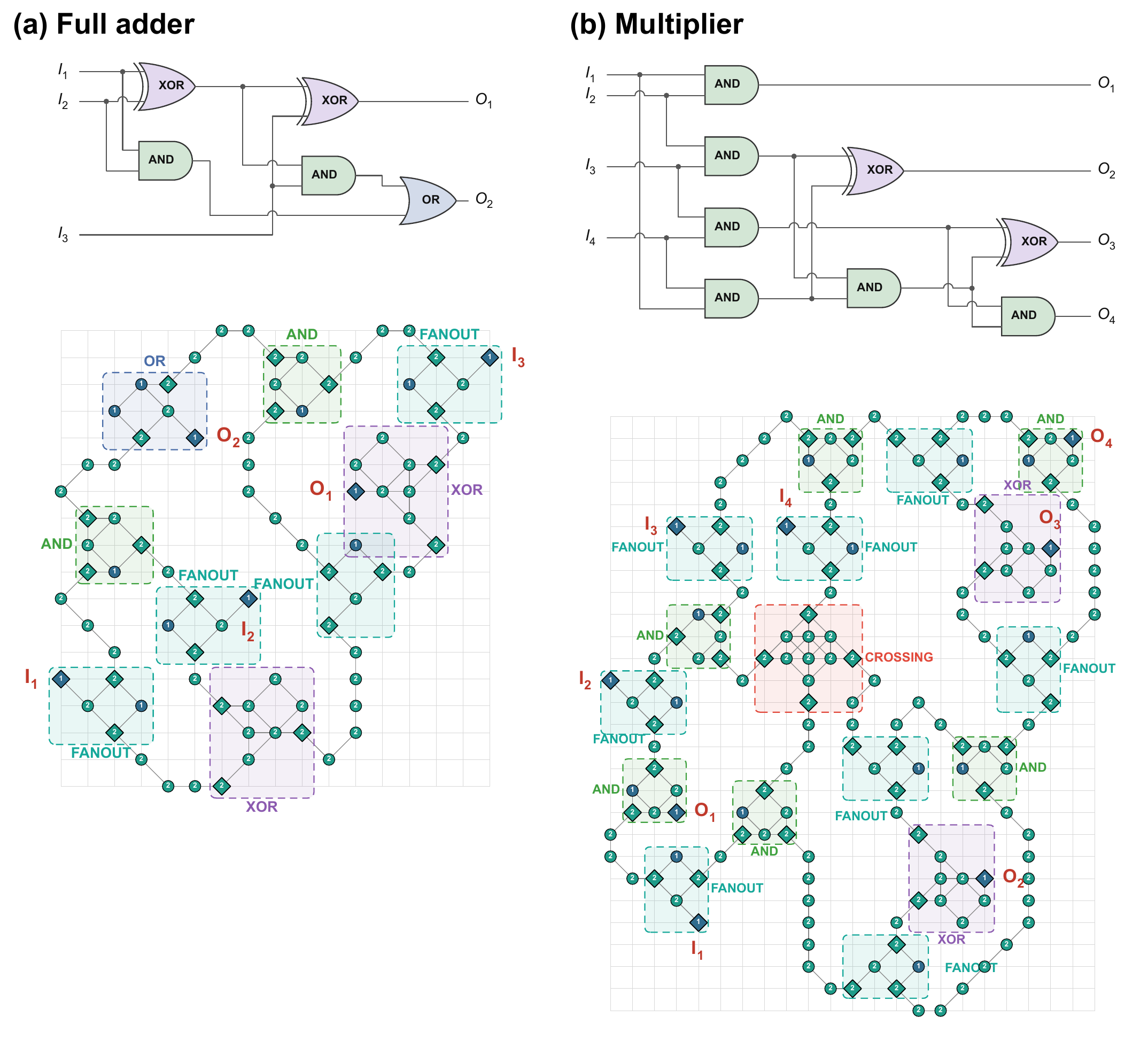}
\caption{Compilation of multi gate arithmetic blocks into weighted
king subgraphs by the automated pipeline.
\textbf{(a)}~One-bit full adder: logical circuit (top) and compiled
$85$-atom weighted KSG (bottom). Dashed boxes outline the constituent
XOR, AND, OR, and fanout gadgets; diamonds mark the gadget I/O pin
atoms, circles the internal and wire atoms; node color encodes the
on-site weight $w \in \{1,2\}$ printed on each atom; and the ports
$I_{1}$--$I_{3}$, $O_{1}$, $O_{2}$ appear as labeled boundary atoms.
The eight degenerate maximum weight independent sets (weight $W^{\star}=73$)
enumerate the full adder truth table.
\textbf{(b)}~Two-bit multiplier: logical circuit (top) and compiled
$165$-atom weighted KSG (bottom), with the same annotation and weight
alphabet; the router additionally inserts a crossing gadget (red box)
where two independent signals must cross, and the sixteen degenerate
maximum weight independent sets (weight $W^{\star}=142$) enumerate the complete
multiplication table.}
\label{fig:arith-blocks}
\end{figure*}
The preceding subsections established the Circuit-SAT program end to
end on hand-checkable instances. The larger circuits behind practical
problems add two requirements: the composed KSG must preserve logical
correctness at the block level, with no spurious configurations
introduced when many gadgets are connected, and the compilation must run
automatically through the placement and routing pipeline of
Sec.~\ref{subsec:method-compilation}. We test both on the one-bit full
adder (five gates) and the two-bit multiplier (eight gates)---the
standard units from which larger arithmetic circuits, and through them
the practically relevant Circuit-SAT instances of
Sec.~\ref{subsec:csat-problem}, are assembled---already rich in
connectivity, signal propagation, and fanout yet small enough for exact
correctness analysis.

A demanding first test is the full adder. With summand
bits $I_{1}, I_{2}$ and carry-in $I_{3}$
[Fig.~\ref{fig:arith-blocks}(a)], the block computes the sum and
carry-out
\begin{equation}
O_{1} \;=\; I_{1} \oplus I_{2} \oplus I_{3},
\qquad
O_{2} \;=\; \bigl( (I_{1} \oplus I_{2}) \wedge I_{3} \bigr)
            \;\vee\; (I_{1} \wedge I_{2}).
\label{eq:full-adder-carry}
\end{equation}
Five gates---two XOR, two AND, and one OR---realize this through the
intermediate signals $s = I_{1} \oplus I_{2}$,
$c_{1} = s \wedge I_{3}$, and $c_{2} = I_{1} \wedge I_{2}$. Because the
carry-out reuses the signals that form the sum, $I_{1}$, $I_{2}$,
$I_{3}$, and $s$ each drive two downstream gates---exercising a deeper
gate stack, four independent fanout points, and the atom-heavy XOR
gadget all at once. The automated pipeline compiles the circuit into the
$85$-atom weighted KSG of Fig.~\ref{fig:arith-blocks}(a): two $9$-atom
XOR, two $6$-atom AND, and one $6$-atom OR gadget, four $5$-atom fanout
gadgets, and $29$ wire atoms, with no crossing gadget required. The wire
rule keeps the composed block on the binary weight
alphabet of the isolated gadgets
(Sec.~\ref{subsec:method-compilation}). Exact enumeration (Supplemental Material,
Sec.~\ref{sm:adder-verification}) confirms the composition gate by gate:
the eight degenerate optima at $W^{\star}=73$, one per input assignment,
project port by port onto the rows of the full adder truth table, with
the correct intermediate signals on the internal gate ports and no
spurious projection at the optimum weight.

In the two-bit multiplier, the four inputs encode the operands
$A = (I_{3}\,I_{1})$ and $B = (I_{4}\,I_{2})$ and the four outputs
$O_{1}$ (least significant) through $O_{4}$ (most significant) encode
the product $A \times B$: four AND gates form the partial products and a
two-stage XOR/AND half-adder cascade sums them. It compiles into the
$165$-atom KSG of Fig.~\ref{fig:arith-blocks}(b), whose wiring is dense
enough that two inter-gadget signal lines cannot be routed in the plane
without intersecting; at that point the router inserts an $11$-atom
crossing gadget (red box), which lets the two independent signals pass
through each other without logical interference. The sixteen degenerate
optima (weight $W^{\star}=142$) reproduce the complete multiplication table
(Supplemental Material, Sec.~\ref{sm:multi-verification}). Both layouts
expose the circuit ports as
labeled boundary atoms, demonstrating that the gate level KSG library
composes correctly across arithmetic blocks of increasing size.

The atom counts of these blocks also quantify, on concrete circuits, the
saving established statistically in Sec.~\ref{sec:resources}. Pushed
through the CNF mediated chain of Sec.~\ref{subsec:csat-to-mis}, the
identical full adder compiles to an unweighted king subgraph of $2563$
atoms and the identical multiplier to one of $5060$ atoms: the Tseitin
encoding alone produces intermediate MIS graphs of $57$ and $98$
vertices, which the wire chain embedding then inflates by more than an
order of magnitude. The direct gadget compilations at $85$ and $165$
atoms are therefore smaller by factors of about $30$ in both cases---the
difference between instances that fit comfortably within present-day
neutral atom arrays of a few hundred to a few thousand
atoms~\cite{Ebadi2021Nature256,Scholl2021AFM2D,Manetsch2025_6100Atoms}
and ones that already strain them at the scale of a one-bit adder.

\section{Conclusion}
\label{sec:conclusion}
We presented CAMERA, a complete approach for solving Circuit-SAT on
neutral atom
hardware: a gate level encoding that turns a Boolean circuit directly
into an MWIS problem on a weighted king
subgraph, and an output branching protocol that extracts the answer. On
random benchmarks of three to eight gates the direct route uses an
average factor of $22.4$ fewer atoms than the conventional CNF mediated
reduction, rising to a factor of about $30$ on the arithmetic
blocks---the difference between instances that fit comfortably in
present-day arrays and ones that strain them at the scale of a one-bit
adder. Every layer is verified constructively: exact ground state
calculations confirm the $30$-atom three-gate instance, the $85$-atom
full adder, and the $165$-atom multiplier against their complete truth
tables; a closed system tensor network simulation of a hardware
compatible annealing schedule reaches the exact ground state manifold in
$58.0\%$ of $10^{3}$ projective samples on the $30$-atom instance
($54.6\%$ on its $27$-atom output branched counterpart); and a minimal
unsatisfiable instance is certified by the one-unit weight deficit of
its branched optimum. Throughout, the binary weight alphabet
$w \in \{1,2\}$ of the gadget library survives composition, so two local
detuning values suffice, and the VLSI inspired placement and routing
compiler assembles the blockade faithful layouts automatically.

A structural feature with broader significance is the dual role of the
compiled king subgraph: with all ports free, its maximum weight
independent sets enumerate the circuit's complete truth table, one
degenerate optimum per input assignment; branching the output port turns
this manifold into a satisfiability decider, and branching the input
ports into a forward evaluator that computes the Boolean function by
relaxing into the energetically preferred many-body configuration rather
than by propagating signals through a temporal gate sequence. Because
branching only deletes vertices---implemented by local detuning changes,
not by re-layout---one compiled array supports a family of tasks
selected at run time: pinning the inputs evaluates the circuit, pinning
the output inverts it, recovering an input assignment consistent with a
demanded output, pinning a mixed subset of ports poses a general
constraint query on the circuit, and sampling the free manifold draws
from the circuit's valid input--output rows. A compiled Rydberg array
thus becomes a reusable logic block for ground state computation,
reprogrammed by choosing which ports to constrain rather than by
redesigning the layout.

As analyzed in Sec.~\ref{sec:resources}, the resource advantage has a
structural origin: the gadgets encode each logic operation directly and
minimally, and the compiler pays routing overhead only for the wiring
the circuit actually contains, so the generically quadratic growth of
the CNF route is replaced by near-linear scaling for the structured
feed-forward circuits of practical interest. The most immediate
consequence of this reduced encoding expense is that compiled
Circuit-SAT instances now fit on today's hardware: the
$30$--$165$-atom instances verified here, requiring only two local
detuning values, can be executed on
current weighted-MIS platforms with local light shift
control~\cite{deOliveira2025WeightedMWIS}, closing the loop between the
theoretical construction and hardware native optimization. Beyond this
experiment, two directions extend naturally: scaling the automated
compilation to industry-standard benchmark families---the combinational
circuits and equivalence checking miters of formal hardware
verification~\cite{Li2005FormalVerification}---would provide a concrete
performance comparison at practically relevant circuit sizes and connect
the solver to established electronic design automation workflows; and
the correspondence between circuit structure and optimization hardness
deserves systematic study, since the gadget-based encoding introduces
structured connectivity that differs qualitatively from the random or
geometric instances of prior MIS hardness
analyses~\cite{Andrist2023Hardness,Cain2023FlatLandscapes}, leaving
open whether this structure makes the encoded problem harder or easier
for quantum annealing dynamics.

\section{ACKNOWLEDGMENTS}
We are deeply grateful to Jin-Guo Liu for many insightful discussions throughout this project and for valuable comments and suggestions on the manuscript. This work is supported by Science, Technology and Innovation Bureau of Shenzhen Municipality.

\bibliography{references}

\clearpage
\onecolumngrid

\begin{center}
\textbf{\large Supplemental Material for ``Encoding Circuit Satisfiability in Rydberg Atom Arrays''}
\end{center}

\setcounter{section}{0}
\setcounter{equation}{0}
\setcounter{figure}{0}
\setcounter{table}{0}
\setcounter{theorem}{0}
\renewcommand{\thesection}{S\Roman{section}}
\renewcommand{\theequation}{S\arabic{equation}}
\renewcommand{\thefigure}{S\arabic{figure}}
\renewcommand{\thetable}{S\arabic{table}}
\renewcommand{\thetheorem}{S\arabic{theorem}}

\section{Composition theory of the gadget encoding}
\label{sm:composition-theory}

This section develops the composition theory behind the two theorems of
the main text: the circuit realization theorem
(Theorem~\ref{thm:realization}) and the output branch SAT criterion
(Theorem~\ref{thm:branch-criterion}). The notation follows the main
text: $W^{\star}(G)$ is the optimal weight of Eq.~\eqref{eq:mwis} on a
weighted graph $(G,w)$; for a vertex set $S$ we write
$w(S)=\sum_{v\in S}w_{v}$ for its weight, the quantity $W(S)$ of
Eq.~\eqref{eq:mwis}, and the local weight functions $w_{i}$ are extended
to sets in the same way. Here $N[v]$ is the closed neighborhood of $v$.
For a graph $G$
and a vertex subset $U$, $G-U$ denotes the induced graph after deleting
$U$.

\subsection{Pin relations and gadget contracts}
\label{sm:pin-relations}

Let $G=(V,E)$ be a finite graph with a positive vertex weight function
$w$. Designate $k$ pairwise distinct vertices of $G$ as \emph{pins} and
record them in the ordered tuple $P=(p_{1},\ldots,p_{k})$; in the
compiled graphs of the main text the pins are the port atoms. For an
independent set $S$ of $G$, membership of a pin encodes its Boolean
value, and the resulting ordered assignment is the pin projection
\begin{equation}
b_{i}(S)=
\begin{cases}
1, & p_{i}\in S,\\
0, & p_{i}\notin S,
\end{cases}
\qquad
\pi_{P}(S)=\bigl(b_{1}(S),\ldots,b_{k}(S)\bigr).
\label{eq:pin-projection}
\end{equation}
Define the set of all maximum weight independent sets of $(G,w)$ by
\begin{equation}
\Opt(G,w)=\arg\max_{S\ \text{independent}}w(S),
\label{eq:generic-optimum-set}
\end{equation}
and the \emph{optimal pin relation} realized by $(G,w)$ as
\begin{equation}
T(G,P)=\{\pi_{P}(S):S\in\Opt(G,w)\}.
\label{eq:optimal-pin-relation}
\end{equation}
A \emph{gadget type} $g$ is specified by an intended Boolean relation
$T_{g}$; a \emph{gadget realization} is a weighted graph $(G_{g},w_{g})$
with pin tuple $P_{g}$, and the realization is correct for type $g$
exactly when it fulfills the contract
\begin{equation}
T(G_{g},P_{g})=T_{g}.
\label{eq:gadget-contract}
\end{equation}

The construction uses seven relation types. The four gate types encode
truth tables: for $P=(A,B,O)$, respectively $P=(A,O)$,
\begin{align}
T_{\mathrm{AND}}&=\{(0,0,0),(0,1,0),(1,0,0),(1,1,1)\},
\label{eq:and-relation}\\
T_{\mathrm{OR}}&=\{(0,0,0),(0,1,1),(1,0,1),(1,1,1)\},
\label{eq:or-relation}\\
T_{\mathrm{XOR}}&=\{(0,0,0),(0,1,1),(1,0,1),(1,1,0)\},
\label{eq:xor-relation}\\
T_{\mathrm{NOT}}&=\{(0,1),(1,0)\}.
\label{eq:not-relation}
\end{align}
The three routing types enforce equality constraints: for
$P_{\mathrm{CROSSING}}=(N,E,S,W)$,
$P_{\mathrm{WIRE}}=(I,O)$, and $P_{\mathrm{FANOUT}}=(I,O_{0},O_{1})$,
\begin{align}
T_{\mathrm{CROSSING}}&=\{(0,0,0,0),(0,1,0,1),(1,0,1,0),(1,1,1,1)\},
\label{eq:crossing-relation}\\
T_{\mathrm{WIRE}}&=\{(0,0),(1,1)\},
\label{eq:wire-relation}\\
T_{\mathrm{FANOUT}}&=\{(0,0,0),(1,1,1)\},
\label{eq:fanout-relation}
\end{align}
so a crossing imposes $n=s$ and $e=w$ independently, a wire imposes
point-to-point equality, and a fanout copies one signal to two outputs.
The fixed AND, OR, NOT, XOR, CROSSING, and FANOUT realizations were
obtained with the automated gadget search tool of
Ref.~\cite{Pan2025TriangularGadgetSearch} and are specified, together
with their contract verifications by exhaustive enumeration, in
Sec.~\ref{sm:gadget-design}.

The WIRE type admits a variable length family, the weighted copy gadget
construction of Ref.~\cite{Nguyen2023ArbitraryConnectivity}. For
$k\in\mathbb{Z}_{\ge 0}$, let $W_{k}$ be the path
$p_{I}-x_{1}-\cdots-x_{2k+1}-p_{O}$ with endpoint pins of weight $1$ and
interior vertices of weight $2$. If $\beta_{k}(a,b)$ denotes the best
independent set weight subject to $\pi_{(p_{I},p_{O})}=(a,b)$, a direct
count of the occupied interior vertices gives
\begin{equation}
\beta_{k}(0,0)=\beta_{k}(1,1)=2k+2,
\qquad
\beta_{k}(0,1)=\beta_{k}(1,0)=2k+1,
\label{eq:wire-branch-values}
\end{equation}
and hence
\begin{equation}
T\bigl(W_{k},(p_{I},p_{O})\bigr)=T_{\mathrm{WIRE}},
\qquad
(2k+2)-\max\{\beta_{k}(0,1),\beta_{k}(1,0)\}=1.
\label{eq:wire-path-contract}
\end{equation}
Every odd-interior path in this weighted family is therefore a
standalone WIRE gadget; Corollary~\ref{cor:wire-insertion} below shows
that inserting such chains into a realization preserves it, which is the
rule applied by the compiler.

\subsection{Circuit realizations and clean composition}
\label{sm:realizations}

Throughout, a circuit is a finite acyclic deterministic Boolean circuit
containing at least one component instance of the seven types above.
Every boundary pin is incident to a primitive component, and every
inter-component wire joins an output pin occurrence to an input pin
occurrence. Wiring is point to point: each inter-component pin
occurrence is incident to at most one wire, and all signal replication
is represented by explicit FANOUT components. For a circuit
$C:\{0,1\}^{n}\to\{0,1\}^{m}$, let
$B_{C}=(I_{1},\ldots,I_{n},O_{1},\ldots,O_{m})$ be its ordered tuple of
distinct boundary pin occurrences. Written in this order, the circuit
truth relation is
\begin{equation}
\mathcal{T}(C)=\left\{\bigl(x,F_{C}(x)\bigr):x\in\{0,1\}^{n}\right\}
\subseteq\{0,1\}^{n+m}.
\label{eq:circuit-truth-relation}
\end{equation}
When a wire is cut to form subcircuits, its source endpoint becomes an
output boundary pin of the upstream subcircuit and its target endpoint
an input boundary pin of the downstream subcircuit.

\begin{definition}[Weighted graph realization of a circuit]
\label{def:circuit-graph-realization}
Let $(G,w)$ be a finite weighted graph and let
$P_{C}=(p_{I_{1}},\ldots,p_{I_{n}},p_{O_{1}},\ldots,p_{O_{m}})$ be an
ordered tuple of pairwise distinct vertices corresponding position by
position to $B_{C}$. The triple $(G,w;P_{C})$ is a \emph{weighted graph
realization} of $C$ if and only if it satisfies both
\begin{equation}
T(G,P_{C})\subseteq\mathcal{T}(C)\quad\text{(soundness)},
\qquad
\mathcal{T}(C)\subseteq T(G,P_{C})\quad\text{(completeness)}.
\label{eq:circuit-graph-contract}
\end{equation}
\end{definition}

Soundness says that every MWIS projects to a valid row of the circuit
truth table; completeness says that every row is the pin projection of
at least one MWIS. Together they are equivalent to
$T(G,P_{C})=\mathcal{T}(C)$, the statement of
Theorem~\ref{thm:realization}. Neither condition asserts a bijection:
several microscopic MWIS may share one boundary projection.

\begin{corollary}[Primitive circuit realizations]
\label{cor:primitive-circuit-realizations}
Each of the six fixed gadget realizations and every WIRE realization
$W_{k}$ is a weighted graph realization of its primitive circuit
component: for the fixed gadgets this is the contract
Eq.~\eqref{eq:gadget-contract} verified in Sec.~\ref{sm:gadget-design},
and for $W_{k}$ it is Eq.~\eqref{eq:wire-path-contract}.
\end{corollary}

For an ordered pin tuple $B$ and an ordered subtuple $A$, write
$B\setminus A$ for the tuple obtained by deleting the entries of $A$
while preserving the order of the remaining entries.

\begin{lemma}[Truth relation composition]
\label{lem:circuit-truth-composition}
Let $C$ be obtained by connecting the ordered distinct output tuple
$O=(o_{1},\ldots,o_{r})$ of $C_{1}$ to the ordered distinct input tuple
$I=(i_{1},\ldots,i_{r})$ of $C_{2}$, where $r\ge 0$. Let $B_{1},B_{2}$
be the local boundary pin orders and $B_{C}$ the external boundary pin
order after the connected coordinates become internal. Put
$E_{1}=B_{1}\setminus O$ and $E_{2}=B_{2}\setminus I$. For bit vectors
$x_{i}$ on $E_{i}$ and an interface vector $a\in\{0,1\}^{r}$, let
$(x_{1},a)_{B_{1}}$ denote the row in $B_{1}$ order carrying $x_{1}$ on
$E_{1}$ and $a$ on $O$, define $(a,x_{2})_{B_{2}}$ analogously, and let
$(x_{1},x_{2})_{B_{C}}$ denote the external row in $B_{C}$ order. Then
\begin{equation}
\mathcal{T}(C)=
\left\{(x_{1},x_{2})_{B_{C}}:
\exists\,a\in\{0,1\}^{r},\
(x_{1},a)_{B_{1}}\in\mathcal{T}(C_{1}),\
(a,x_{2})_{B_{2}}\in\mathcal{T}(C_{2})\right\}.
\label{eq:circuit-truth-composition}
\end{equation}
\end{lemma}

The displayed equality is relational composition along the connected
interface: the connected pins share one interface value, and that value
is existentially hidden once the interface becomes
internal~\cite{Hutton1993RubyInterpreter,Spivak2013WiringDiagrams}.
Every circuit evaluation represented in $\mathcal{T}(C)$ induces two
local rows with the same interface vector $a$; conversely, every such
pair determines a row of $\mathcal{T}(C)$ after $a$ is hidden. When
$r=0$ the formula reduces to the Cartesian product for parallel
composition.

\begin{definition}[Circuit-induced clean weighted pin amalgamation]
\label{def:pin-amalgamation}
Let $C$ be obtained from $C_{1}$ and $C_{2}$ by the connection of
Lemma~\ref{lem:circuit-truth-composition}, along the interface tuples
$O=(o_{1},\ldots,o_{r})$ and $I=(i_{1},\ldots,i_{r})$. For
$i\in\{1,2\}$, let $(G_{i},w_{i};P_{i})$ be a weighted graph
realization of $C_{i}$, with vertex-disjoint copies taken if necessary.
The composite triple $(G,w;P_{C})$ is constructed as follows.

\emph{Graph.} Put
$\mathbf{p}_{1}=(p_{o_{1}},\ldots,p_{o_{r}})$ and
$\mathbf{p}_{2}=(p_{i_{1}},\ldots,p_{i_{r}})$. On the disjoint union
$U=V(G_{1})\sqcup V(G_{2})$, let $\sim$ be the equivalence relation
generated only by $p_{o_{j}}\sim p_{i_{j}}$ for $1\le j\le r$, let
$q:U\to U/{\sim}$ be the quotient map, and put
$\bar{p}_{j}=q(p_{o_{j}})=q(p_{i_{j}})$ and
$\bar{\mathbf{p}}=(\bar{p}_{1},\ldots,\bar{p}_{r})$. Define
\begin{align}
V(G)&=U/{\sim},
\label{eq:amalgamated-vertices}\\
E(G)&=\bigl\{\{q(u),q(v)\}:\{u,v\}\in E(G_{1})\cup E(G_{2}),\
q(u)\ne q(v)\bigr\}.
\label{eq:amalgamated-edges}
\end{align}
The \emph{clean edge condition} comprises the two structural
requirements built into this construction: only the specified pin pairs
are identified, and the edge set contains exactly the images of the two
local edge sets, so no additional edge is introduced.

\emph{Weights.} For every $\xi\in V(G)$,
\begin{equation}
w(\xi)=\sum_{i=1}^{2}\ \sum_{\substack{u\in V(G_{i})\\ q(u)=\xi}}
w_{i}(u),
\label{eq:amalgamated-weight}
\end{equation}
so every nonshared vertex retains its local weight and each identified
pin carries the sum of its two local weights.

\emph{Pins.} The quotient map carries every local boundary pin
correspondence into the composite graph, and the external readout tuple
is $P_{C}=(q(p_{b}))_{b\in B_{C}}$, listed in the order of $B_{C}$; its
entries remain pairwise distinct. For $r=0$ the construction is the
weighted disjoint union.
\end{definition}

\subsection{Composition theorems}
\label{sm:composition-theorems}

\begin{theorem}[Optimal weight additivity]
\label{thm:optimal-weight-additivity}
Let $(G,w;P_{C})$ be the amalgamation of
Definition~\ref{def:pin-amalgamation}. Then
\[
W^{\star}(G)=W^{\star}(G_{1})+W^{\star}(G_{2}).
\]
\end{theorem}

\begin{proof}
Put $W_{i}^{\star}=W^{\star}(G_{i})$. For $a\in\{0,1\}^{r}$, let
$\mathcal{F}_{i}(a)$ be the independent sets of $G_{i}$ with
$\pi_{\mathbf{p}_{i}}=a$, let $\mathcal{F}(a)$ be the independent sets
of $G$ with $\pi_{\bar{\mathbf{p}}}=a$, and let
$W_{i}^{(a)}$ and $B_{a}$ be the corresponding constrained optima, with
value $-\infty$ when a family is empty.

Fix $a$. There is a bijection
$\Phi_{a}:\mathcal{F}_{1}(a)\times\mathcal{F}_{2}(a)\to\mathcal{F}(a)$,
$\Phi_{a}(R_{1},R_{2})=q[R_{1}\cup R_{2}]$. Indeed, the common
interface vector means that for every $j$ either both local
representatives $p_{o_{j}},p_{i_{j}}$ are selected or neither is, and
the clean edge condition ensures that the quotient union is
independent: all local edges are respected and there is no cross graph
edge. Conversely every $R\in\mathcal{F}(a)$ restricts to
$R_{i}=q^{-1}(R)\cap V(G_{i})\in\mathcal{F}_{i}(a)$ with
$R=\Phi_{a}(R_{1},R_{2})$. By Eq.~\eqref{eq:amalgamated-weight},
$w(\Phi_{a}(R_{1},R_{2}))=w_{1}(R_{1})+w_{2}(R_{2})$, and maximizing
over the product family gives
$B_{a}=W_{1}^{(a)}+W_{2}^{(a)}$.

Choose $\widehat{S}_{1}\in\Opt(G_{1},w_{1})$ and put
$a^{\star}=\pi_{\mathbf{p}_{1}}(\widehat{S}_{1})$, so
$W_{1}^{(a^{\star})}=W_{1}^{\star}$. Every value of the input tuple $I$
extends to a complete input vector of $C_{2}$, and the resulting truth
table row is realized by an MWIS of $G_{2}$ by local completeness;
hence $W_{2}^{(a^{\star})}=W_{2}^{\star}$ and
$B_{a^{\star}}=W_{1}^{\star}+W_{2}^{\star}$. For every $a$ the reverse
inequality $B_{a}\le W_{1}^{\star}+W_{2}^{\star}$ follows from
$W_{i}^{(a)}\le W_{i}^{\star}$. Therefore
$W^{\star}(G)=\max_{a}B_{a}=W_{1}^{\star}+W_{2}^{\star}$, and every
composite optimum has this weight while every local optimum has weight
$W_{i}^{\star}$. For $r=0$ the argument reduces to additivity for a
weighted disjoint union.
\end{proof}

For any number $r\ge 0$ of simultaneous connections, the branching
variable is the entire interface vector rather than its coordinates
separately.

\begin{theorem}[Multi-pin branching composition]
\label{thm:multipin-composition}
Let $(G_{i},w_{i};P_{i})$ be a weighted graph realization of $C_{i}$
for $i\in\{1,2\}$, connect the tuple $O$ of $r\ge 0$ distinct outputs
of $C_{1}$ to the tuple $I$ of $r$ distinct inputs of $C_{2}$, and let
$(G,w;P_{C})$ be the amalgamation of
Definition~\ref{def:pin-amalgamation}. Then
\begin{equation}
T(G,P_{C})=\mathcal{T}(C),
\label{eq:multipin-composite-relation}
\end{equation}
so $(G,w;P_{C})$ is a weighted graph realization of $C$.
\end{theorem}

\begin{proof}
For $a\in\{0,1\}^{r}$, define the constrained families
$\mathcal{I}_{i}^{(a)}=\{S\ \text{independent in}\
G_{i}:\pi_{\mathbf{p}_{i}}(S)=a\}$ and
$\mathcal{I}^{(a)}(G)=\{S\ \text{independent in}\
G:\pi_{\bar{\mathbf{p}}}(S)=a\}$, with branch values
\begin{equation}
W_{i}^{(a)}=\max\{w_{i}(S):S\in\mathcal{I}_{i}^{(a)}\},
\qquad
B_{a}=\max\{w(S):S\in\mathcal{I}^{(a)}(G)\},
\label{eq:interface-branch-value}
\end{equation}
equal to $-\infty$ on empty families, and put
$W_{i}^{\star}=W^{\star}(G_{i})$. Let
$\Opt_{i}^{(a)}=\{S\in\Opt(G_{i},w_{i}):\pi_{\mathbf{p}_{i}}(S)=a\}$
and $\mathcal{A}_{i}=\{a:\Opt_{i}^{(a)}\ne\varnothing\}$. Local circuit
realization identifies $\mathcal{A}_{i}$ with the set of complete
interface vectors occurring in $\mathcal{T}(C_{i})$; because $I$
consists of input boundary pins of $C_{2}$,
$\mathcal{A}_{2}=\{0,1\}^{r}$ and
$\mathcal{A}:=\mathcal{A}_{1}\cap\mathcal{A}_{2}=\mathcal{A}_{1}
\ne\varnothing$.

In a fixed branch $a$, cleanliness gives the bijection
\begin{equation}
\mathcal{I}^{(a)}(G)\cong
\mathcal{I}_{1}^{(a)}\times\mathcal{I}_{2}^{(a)},
\qquad
(S_{1},S_{2})\longmapsto q[S_{1}\cup S_{2}],
\label{eq:vector-branch-bijection}
\end{equation}
and summing Eq.~\eqref{eq:amalgamated-weight} over corresponding sets
gives
\begin{equation}
B_{a}=W_{1}^{(a)}+W_{2}^{(a)}.
\label{eq:vector-composite-branch-value}
\end{equation}
The equality $W_{i}^{(a)}=W_{i}^{\star}$ holds exactly for
$a\in\mathcal{A}_{i}$, so $B_{a}=W_{1}^{\star}+W_{2}^{\star}$ if and
only if $a\in\mathcal{A}$, and
\begin{equation}
W^{\star}(G)=W_{1}^{\star}+W_{2}^{\star},
\qquad
\Opt(G,w)=\bigsqcup_{a\in\mathcal{A}}
\bigl\{q[S_{1}\cup S_{2}]:(S_{1},S_{2})\in
\Opt_{1}^{(a)}\times\Opt_{2}^{(a)}\bigr\}.
\label{eq:multipin-optimizer-decomposition}
\end{equation}

For soundness, every composite MWIS in
Eq.~\eqref{eq:multipin-optimizer-decomposition} restricts to two local
MWIS whose truth table rows agree on the interface vector $a$, and
Lemma~\ref{lem:circuit-truth-composition} places their combined
external row in $\mathcal{T}(C)$. For completeness, the same lemma
supplies, for every row of $\mathcal{T}(C)$, two local rows and a
common interface vector $a$; local completeness realizes those rows by
MWIS in $\Opt_{i}^{(a)}$, and their amalgamation attains
$W_{1}^{\star}+W_{2}^{\star}$ and realizes the required external row.
When $r=0$ the proof reduces to additivity and Cartesian product truth
relations for a weighted disjoint union.
\end{proof}

It is essential that Theorem~\ref{thm:multipin-composition} branches on
the complete interface vector. For example, the relations $\{00,11\}$
and $\{01,10\}$ both support $0$ and $1$ in each coordinate separately,
yet share no common two-bit vector, so coordinatewise support cannot
replace equality of the complete vector.

\begin{lemma}[Topological circuit decomposition]
\label{lem:topological-circuit-decomposition}
If a circuit $C$ satisfying the standing assumptions has more than one
primitive component instance, it can be expressed as two proper
subcircuits $C_{1},C_{2}$ joined by $r\ge 0$ distinct output-input pin
pairs.
\end{lemma}

\begin{proof}
Let $H_{C}$ be the DAG whose vertices are primitive component
instances and whose arcs are inter-component wires. If its underlying
undirected graph is disconnected, a nonempty union of connected
components and its complement define proper subcircuits with $r=0$.
Otherwise choose a topological ordering
$g_{1},\ldots,g_{N}$~\cite{Kahn1962TopologicalSorting} and fix any
$1\le k<N$; let $C_{1}$ contain the prefix $\{g_{1},\ldots,g_{k}\}$ and
$C_{2}$ the remainder. Topological order forbids arcs from suffix to
prefix, so every crossing arc runs from an output of $C_{1}$ to an
input of $C_{2}$; cutting these $r\ge 1$ arcs creates matching boundary
pin pairs, and the point-to-point wiring assumption keeps the resulting
pin occurrences pairwise distinct.
\end{proof}

\begin{theorem}[Finite circuit relation composition]
\label{thm:relation-composition}
Every finite acyclic circuit satisfying the standing assumptions and
assembled from the seven primitive components admits a weighted graph
realization $(G_{C},w_{C};P_{C})$ obtained from the primitive gadget
realizations by finitely many clean weighted pin amalgamations,
\begin{equation}
T(G_{C},P_{C})=\mathcal{T}(C).
\label{eq:finite-circuit-realization}
\end{equation}
\end{theorem}

\begin{proof}
Induct on the number $N$ of primitive component instances. For $N=1$
the claim is Corollary~\ref{cor:primitive-circuit-realizations}. For
$N>1$, Lemma~\ref{lem:topological-circuit-decomposition} expresses $C$
as two proper subcircuits joined along $r\ge 0$ output-input pairs;
both have realizations by the induction hypothesis, and
Theorem~\ref{thm:multipin-composition} applied to their amalgamation
produces a realization of $C$.
\end{proof}

\begin{corollary}[WIRE insertion]
\label{cor:wire-insertion}
Let $C^{\mathrm{route}}$ be obtained from $C$ by replacing any finite
collection of point-to-point connections by finite serial chains of
WIRE components. Then, in the original external boundary order,
\begin{equation}
\mathcal{T}(C^{\mathrm{route}})=\mathcal{T}(C),
\label{eq:wire-insertion-relation}
\end{equation}
and if every inserted WIRE is realized by a graph $W_{k}$ satisfying
Eq.~\eqref{eq:wire-path-contract}, with all endpoint identifications
clean amalgamations, the composed weighted graph is a weighted graph
realization of $C$.
\end{corollary}

\begin{proof}
Each WIRE relation forces its two endpoint signals to coincide, so
applying Lemma~\ref{lem:circuit-truth-composition} successively along a
serial chain existentially hides the internal WIRE pins while
preserving equality of the outer endpoints, which proves
Eq.~\eqref{eq:wire-insertion-relation}; the graph statement follows
from repeated applications of
Theorem~\ref{thm:multipin-composition}.
\end{proof}

\begin{proof}[Proof of Theorem~\ref{thm:realization}]
The compiled graph $G_{C}$ is assembled from the primitive gadget
realizations by clean amalgamations, with routed connections realized
as WIRE chains. The relation equality
$T(G_{C},P_{C})=\mathcal{T}(C)$ is
Theorem~\ref{thm:relation-composition} combined with
Corollary~\ref{cor:wire-insertion}. The additivity of the optimal
weight over the constituent gadgets and wire chains follows by
applying Theorem~\ref{thm:optimal-weight-additivity} at every
amalgamation step of the induction.
\end{proof}

\subsection{The output branch criterion}
\label{sm:output-branch}

Let $(G_{C},w_{C};P_{C})$ be a compiled weighted graph realization with
output port atom $v_{O}$ of weight $w_{O}$, and let
$W_{\mathrm{br}}=W^{\star}(G_{C}-N[v_{O}])$ as in the main text. Write
$\mathcal{G}$ for the gadget instances placed by the compiler and
$\mathcal{W}$ for the wire chains it routes, so that the additivity
clause of Theorem~\ref{thm:realization} reads
$W^{\star}=\sum_{g\in\mathcal{G}}C_{g}+\sum_{c\in\mathcal{W}}(m_{c}+1)$,
with $C_{g}$ the gadget optimum and $m_{c}$ the number of interior atoms
of chain $c$, both tabulated in Table~\ref{tab:sm-gadget-summary}.

\begin{proof}[Proof of Theorem~\ref{thm:branch-criterion}]
The independent sets of $G_{C}$ containing $v_{O}$ are in bijection
with the independent sets of $G_{C}-N[v_{O}]$ through
$S_{\mathrm{br}}\mapsto S_{\mathrm{br}}\cup\{v_{O}\}$, so their best
possible weight is $w_{O}+W_{\mathrm{br}}$, the first argument of
Eq.~\eqref{eq:branch}. The equality
$W^{\star}-W_{\mathrm{br}}=w_{O}$ therefore holds exactly when some
global MWIS contains $v_{O}$. By Theorem~\ref{thm:realization} this is
equivalent to the existence of an input $x\in\{0,1\}^{n}$ with
$F_{C}(x)=1$, i.e.\ to satisfiability. In the equality case, for any
MWIS $S_{\mathrm{br}}$ of the branched graph the set
$S_{\mathrm{br}}\cup\{v_{O}\}$ is a full optimum, and by soundness its
port projection is a row of the truth table with output $1$, so the
input port occupations spell a satisfying assignment. Under a strict
deficit no global MWIS selects $v_{O}$, so no gate consistent
configuration attains $O=1$ and the instance is unsatisfiable.
\end{proof}

The input branching variant of the same identity, applied port by port
to a specified input assignment, underlies the forward evaluation
procedure of Sec.~\ref{subsec:method-forward}.

\begin{remark}[Exact certificate versus sampling evidence]
\label{rem:certificate}
Equation~\eqref{eq:branch-criterion} is a mathematical decision
certificate only when $W_{\mathrm{br}}$ is a certified optimal value;
$W^{\star}$ is exact by Theorem~\ref{thm:realization} and carries no such
requirement. A sampled state that decodes to an input $x$ with
$F_{C}(x)=1$ is a constructive SAT witness even if global optimality is
unknown. By contrast, failure to observe the equality in finitely many
analog samples supplies lower bounds and empirical evidence, not an
UNSAT certificate.
\end{remark}

\begin{center}
\noindent\rule{\columnwidth}{0.4pt}\\[1pt]
\noindent\textbf{Algorithm S1.}~Circuit-SAT decision and witness
extraction.\\[-2pt]
\noindent\rule{\columnwidth}{0.4pt}
\end{center}
\vspace{-6pt}
\begin{algorithmic}[1]
\Require circuit $C$ with designated output $O$
\State $(G_{C},w,v_{O},\mathcal{G},\mathcal{W})\gets
  \Call{CompileToWeightedKSG}{C}$
  \Comment{Sec.~\ref{subsec:method-compilation}}
\State $W^{\star}\gets\sum_{g\in\mathcal{G}}C_{g}
  +\sum_{c\in\mathcal{W}}(m_{c}+1)$
  \Comment{Theorem~\ref{thm:realization}; no optimization}
\State $(W_{\mathrm{br}},S_{\mathrm{br}})\gets
  \Call{ExactMWIS}{G_{C}-N[v_{O}],w}$
\If{$W^{\star}-W_{\mathrm{br}}=w_{O}$}
  \State \Return \textsc{sat}, $\Call{ReadInputPorts}{S_{\mathrm{br}}}$
\Else
  \State \Return \textsc{unsat}
\EndIf
\end{algorithmic}
\vspace{-4pt}
\noindent\rule{\columnwidth}{0.4pt}
\vspace{4pt}

\section{Design of the weighted KSG gate gadgets}
\label{sm:gadget-design}

This section documents the design, search procedure, and explicit verification
of the weighted king-subgraph (KSG) gate gadgets---AND, NOT, OR, and
XOR, of which the three two-input gates constitute the gate library shown
in Fig.~\ref{fig:overview}(b) of
the main text---together with the two auxiliary gadgets used by the
placement-and-routing compiler: the wire-crossing gadget, which lets two
routed signals cross without interacting, and the fan-out gadget,
which duplicates a signal for multiple downstream gates. For each gadget we
specify the atom coordinates on the king grid, the integer on-site weights,
and the complete set of degenerate maximum-weight independent sets (MWIS),
which realises exactly the truth table of the corresponding Boolean
primitive. The library obeys a uniform selection criterion: every atom
weight is $1$ or $2$, with all ports at weight $1$, so the analog
implementation requires only two local-detuning values and ports look
identical across all six gadgets. Each gadget is either certified minimal
by exhaustive enumeration or verified minimal within its search class, as
detailed below.

\subsection{Search formulation}
\label{sm:search-formulation}

A gate gadget is an induced subgraph $G=(V,E)$ of the king grid: vertices
occupy integer grid sites, and two vertices are adjacent whenever their
Chebyshev distance is one, so that horizontal, vertical, and diagonal
nearest neighbours are connected. A subset $P=(p_{1},\dots,p_{k})\subset V$
of \emph{pins} carries the Boolean interface; for the two-input gates below,
$p_{1},p_{2}$ store the inputs $x_{1},x_{2}$ and $p_{3}$ stores the output
$y$. The pins must be pairwise non-adjacent, since truth-table rows in which
several pins are simultaneously $1$ must be realisable as independent sets.

Following the blockade physics of Sec.~\ref{subsec:rydberg-arrays}, the
admissible configurations of the gadget are the maximal independent sets
of $G$. Each vertex $v$ carries an integer weight $w_{v}\ge 1$ (mapped to a
local detuning in the analog implementation), and a configuration $S$ scores
$w(S)=\sum_{v\in S}w_{v}$. Writing $\pi(S)\in\{0,1\}^{k}$ for the bit pattern
that $S$ induces on the pins, the gadget is a valid encoding of a truth
table $T\subset\{0,1\}^{k}$ when the ground-state manifold reads out exactly
the allowed rows:
\begin{equation}
\{\,\pi(S) \;:\; S \text{ maximal independent},\; w(S)=C\,\} = T,
\qquad
C=\max_{S}\,w(S).
\label{eq:sm-gadget-condition}
\end{equation}
The weight search is an integer program in the spirit of the automated
gadget-search strategy of Ref.~\cite{Pan2025TriangularGadgetSearch}, adapted
here to the king-subgraph geometry: for each candidate graph and pin
placement, the maximal independent sets are grouped by their pin patterns,
one representative configuration is selected for each allowed row of $T$,
and the solver searches for integer weights such that the representatives
tie at a common score $C$ while every other maximal independent set scores
at most $C-1$,
\begin{equation}
w(S_{r})=C \quad \forall\, r\in T,
\qquad
w(S)\le C-1 \quad \text{for every other maximal independent set } S.
\label{eq:sm-ilp}
\end{equation}
Two further ingredients complete the formulation. First, the physical
implementation imposes a geometric \emph{port constraint}: each input and
output port should be king-adjacent to at most two auxiliary atoms, so that
ports remain accessible for inter-gadget wiring without exceeding the
nearest-neighbour budget of the Rydberg layout. Second, candidate graphs
are enumerated at two levels of rigour. Exploratory searches enumerate
connected induced subgraphs of small king grids in order of increasing
vertex count, so a successful gadget is accompanied by the negative results
at all smaller sizes within the explored class. The optimality certificates
quoted below rest instead on exhaustive sweeps with no grid box: layouts
are enumerated as \emph{polyplets}---king-connected cell sets normalised by
translation, grown size by size---so every king-subgraph layout of a given
size is covered, and the per-size layout counts ($110$, $638$, $3832$,
$23592$, $147941$, and $940982$ for sizes $4$--$9$) reproduce the known
census of king-connected polyominoes, independently confirming that the
enumeration is complete. Nonexistence claims are established with the
weight solver in its most permissive form (integer weights unbounded above,
auxiliary weights allowed to vanish, no port-degree bound), so they hold a
fortiori under the library conventions. Each reported gadget is finally
verified independently of the solver by brute-force enumeration of all
maximal independent sets, confirming
condition~\eqref{eq:sm-gadget-condition}.

\subsection{AND gadget}
\label{sm:and-gadget}

The AND gadget uses six atoms on a $3\times3$ king grid: pins
\begin{equation}
x_{1}=(1,1),\qquad x_{2}=(3,2),\qquad y=(1,3),
\end{equation}
and three auxiliary atoms $a=(2,1)$, $b=(1,2)$, $c=(2,3)$, with the king
edges $(x_{1},a)$, $(x_{1},b)$, $(x_{2},a)$, $(x_{2},c)$, $(y,b)$, $(y,c)$,
$(a,b)$, $(b,c)$. The solver returns the weights
\begin{equation}
(w_{x_{1}},w_{x_{2}},w_{y},w_{a},w_{b},w_{c})=(1,1,1,1,2,2).
\end{equation}
With these weights the maximum-weight value is $C=3$, attained by exactly
four maximal independent sets, one per row of the AND truth table
[Fig.~\ref{fig:sm-and}]:
\begin{equation}
000:\{a,c\},\qquad
100:\{x_{1},c\},\qquad
010:\{x_{2},b\},\qquad
111:\{x_{1},x_{2},y\}.
\end{equation}
The strongest invalid pin pattern, $001:\{y,a\}$, scores only $2$ and is
therefore excluded from the ground-state manifold. The optimal pin patterns
are exactly $\{000,100,010,111\}$, i.e.\ $y=x_{1}\wedge x_{2}$.

We note that the weighting is essential: treated as an \emph{unweighted}
maximum independent set problem, the same graph has the unique maximum
independent set $\{x_{1},x_{2},y\}$ and records only the $111$ row. The full
AND relation emerges only under the weighted rule of
Eq.~\eqref{eq:sm-gadget-condition}, which compares the weights of maximal
independent sets.

The gadget is certified optimal on all counts by the exhaustive polyplet
sweep of Sec.~\ref{sm:search-formulation}. No king-connected layout of four
or five cells admits AND even under the most permissive weight rules, so
six atoms is a global minimum; and no six-atom AND gadget fits a bounding
box smaller than $3\times3$, so the footprint is minimal as well. Within
the $3\times3$ box the minimum total weight is $8$, attained uniquely (up
to the eight lattice symmetries) by the gadget above, whose port
king-degrees are all within the wiring budget of two. The only strictly
lighter six-atom gadget has total weight $7$ but requires a $4\times4$
bounding box and a degree-$3$ output port, violating that budget; under
the port constraint the gadget above is the unique weight-optimal AND
gadget.

\begin{figure}[t]
\centering
\includegraphics[width=0.98\linewidth]{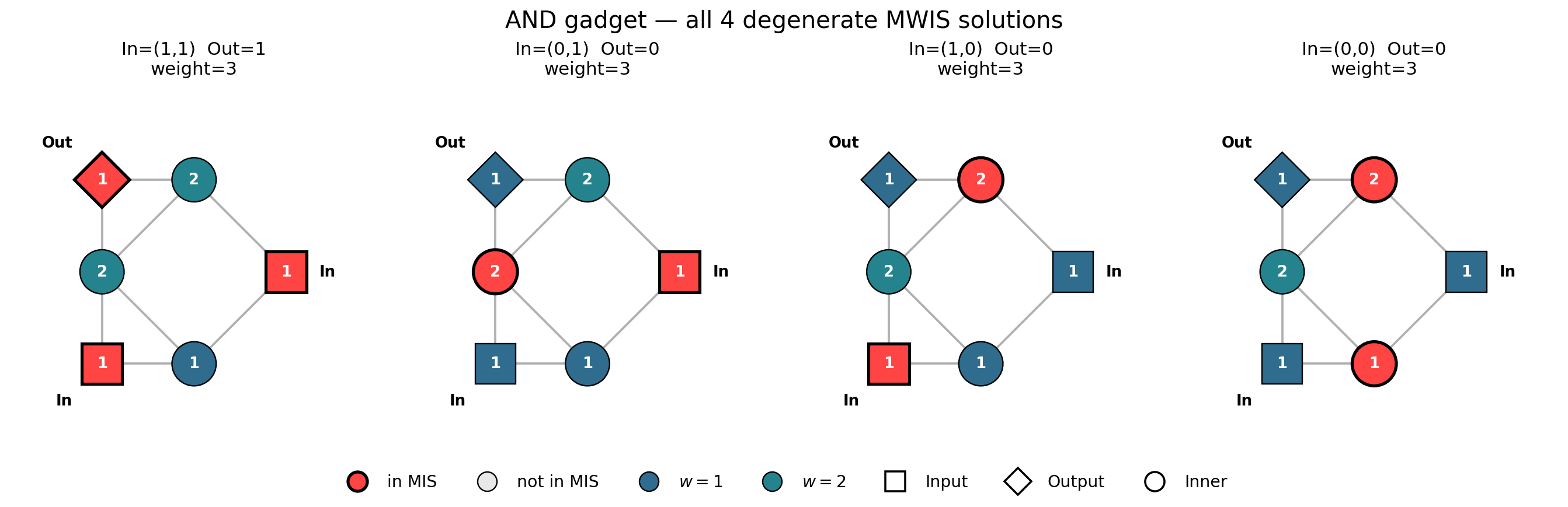}
\caption{The six-atom AND gadget and its four degenerate MWIS solutions, one
per row of the AND truth table. Squares are the input ports
($x_{1},x_{2}$), the diamond is the output port ($y$), and circles are inner
auxiliary atoms; the number inside each vertex is its integer weight. Red
fill marks atoms in the selected independent set; atoms not in the set are
coloured by their weight (legend). All four
rows $000$, $100$, $010$, $111$ tie at the maximum weight $C=3$, so the
ground-state pin patterns are exactly the AND truth table. Six atoms, the
$3\times3$ footprint, and the total weight are all certified minimal.}
\label{fig:sm-and}
\end{figure}

\FloatBarrier
\subsection{NOT gadget}
\label{sm:not-gadget}

The NOT gate ($y=\neg x$) is the one primitive that requires no auxiliary
atom at all: the gadget is the single king edge $K_{2}$, two atoms
\begin{equation}
x=(1,1),\qquad y=(2,1),
\end{equation}
placed king-adjacent, with equal unit weights
\begin{equation}
(w_{x},w_{y})=(1,1).
\end{equation}
The maximal independent sets of $K_{2}$ are exactly $\{x\}$ and $\{y\}$,
with pin patterns [Fig.~\ref{fig:sm-not}]
\begin{equation}
10:\{x\}\ (\text{weight }1),
\qquad
01:\{y\}\ (\text{weight }1),
\end{equation}
both tying at $C=1$, so the ground-state pin patterns are exactly
$\{10,01\}$, i.e.\ $y=\neg x$. The two invalid rows never even enter the
state space: $11$ is not an independent set, since the edge forbids it, and
$00$---the empty set---is not maximal, since either atom can always be
added. The maximal-independent-set structure of $K_{2}$ therefore \emph{is}
the NOT truth table, the same perfect property as the OR gadget below, with
nothing left for the weights to suppress.

Three remarks are in order. First, the pins must be adjacent, which
inverts the usual pin rule: the other gate gadgets keep their pins pairwise
non-adjacent because some truth-table row switches all pins on
simultaneously, but NOT has no row with both pins on, and the edge itself
is the entire mechanism---it enforces the anti-correlation. Placed
non-adjacent, the two atoms would have the single maximal independent set
$\{x,y\}$, reading the wrong pattern $11$. Second, the weights must be
equal: a scan over $(w_{x},w_{y})\in\{1,2,3\}^{2}$ confirms that the two
rows tie exactly when $w_{x}=w_{y}$, and unit weights suffice, so NOT costs
no weight-$2$ atom. Third, the gadget is trivially optimal in every metric
of Table~\ref{tab:sm-gadget-summary}: two atoms is the minimum conceivable
for a two-pin relation, the bounding box is $2\times1$, all weights are
$1$, and each port has zero auxiliary neighbours, so the port constraint is
satisfied vacuously. Since the encoded relation $\{01,10\}$ is symmetric in
$x$ and $y$ (it is $x\oplus y=1$), either atom can serve as the input: the
gadget is its own mirror image and doubles as an inverting wire segment.

\begin{figure}[t]
\centering
\includegraphics[width=0.40\linewidth]{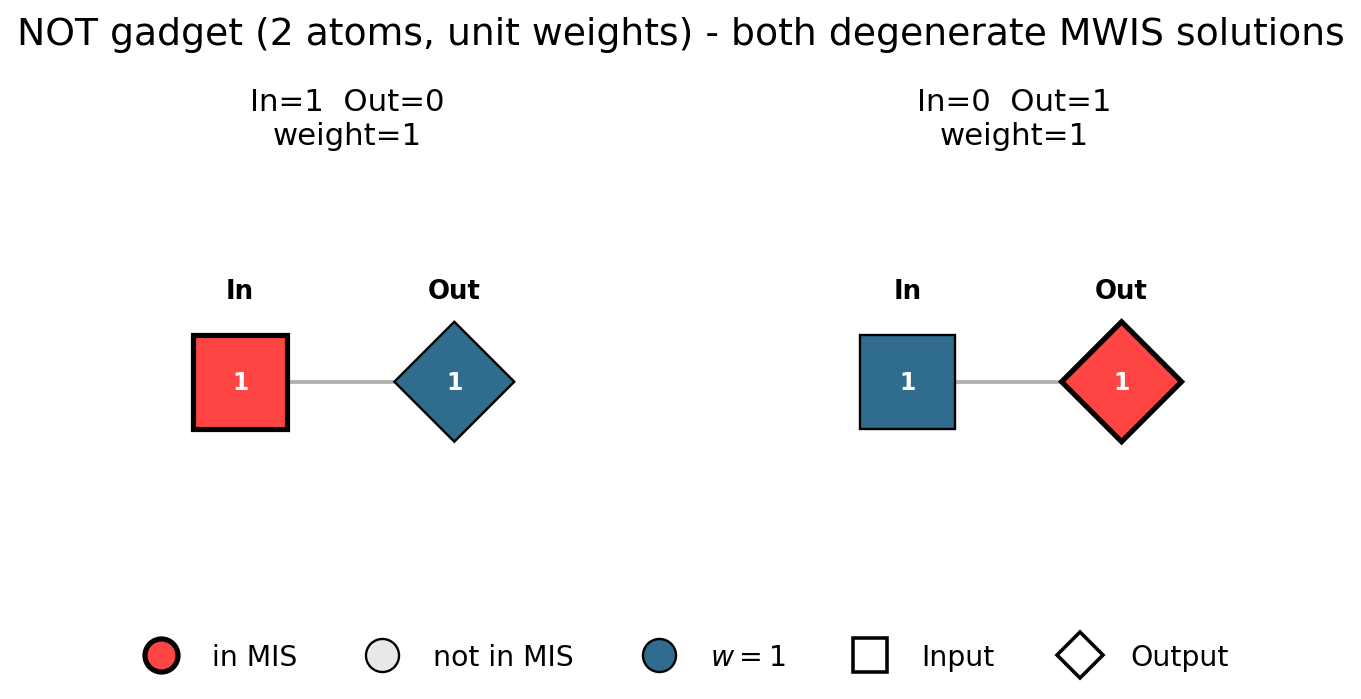}
\caption{The two-atom NOT gadget and its two degenerate MWIS solutions,
one per row of the NOT truth table, in the same convention as
Fig.~\ref{fig:sm-and}. The square is the input port ($x$), the diamond is
the output port ($y$), and the number inside each vertex is its integer
weight. Both rows $10$ and $01$ tie at the maximum weight $C=1$, while the
invalid rows $11$ (not independent) and $00$ (not maximal) never appear in
the state space at all, so the gadget needs no auxiliary atom and no
weight above $1$.}
\label{fig:sm-not}
\end{figure}

\FloatBarrier
\subsection{OR gadget}
\label{sm:or-gadget}

The OR gadget likewise uses six atoms, on a $4\times3$ window of the king
grid: pins
\begin{equation}
x_{1}=(2,3),\qquad x_{2}=(3,1),\qquad y=(4,3),
\end{equation}
and auxiliary atoms $a=(1,2)$, $b=(2,1)$, $c=(3,2)$, with the king edges
$(x_{1},a)$, $(x_{1},c)$, $(x_{2},b)$, $(x_{2},c)$, $(y,c)$, $(a,b)$,
$(b,c)$. The solver returns the weights
\begin{equation}
(w_{x_{1}},w_{x_{2}},w_{y},w_{a},w_{b},w_{c})=(1,1,1,1,1,2).
\end{equation}
The maximum-weight value is $C=3$, attained by exactly four maximal
independent sets [Fig.~\ref{fig:sm-or}]:
\begin{equation}
000:\{a,c\},\qquad
101:\{x_{1},y,b\},\qquad
011:\{x_{2},y,a\},\qquad
111:\{x_{1},x_{2},y\},
\end{equation}
so the ground-state pin patterns are exactly $\{000,101,011,111\}$, i.e.\
$y=x_{1}\vee x_{2}$. The gadget has a distinctive structural property: the
four optimal configurations are the \emph{only} maximal independent sets of
the graph, so there is no spurious pattern to suppress at all---the
maximal-independent-set structure of the graph \emph{is} the OR truth
table.

The exhaustive polyplet sweep certifies this gadget in all three respects.
No king-connected layout of four or five cells admits OR even under the
most permissive weight rules, so six atoms is a global minimum. At six
atoms no all-unit-weight gadget exists, and no gadget is lighter than this
one: the maximum weight $2$ and the total weight $7$ are both globally
minimal. The layout requires a $4\times3$ bounding box---searches confined
to the $3\times3$ grid instead find a six-atom OR variant whose off-state
representative is a single weight-$3$ auxiliary atom, with weights
$(1,1,1,3,1,1)$---and its port king-degrees are $(2,2,1)$, within the
wiring budget.

\begin{figure}[t]
\centering
\includegraphics[width=0.98\linewidth]{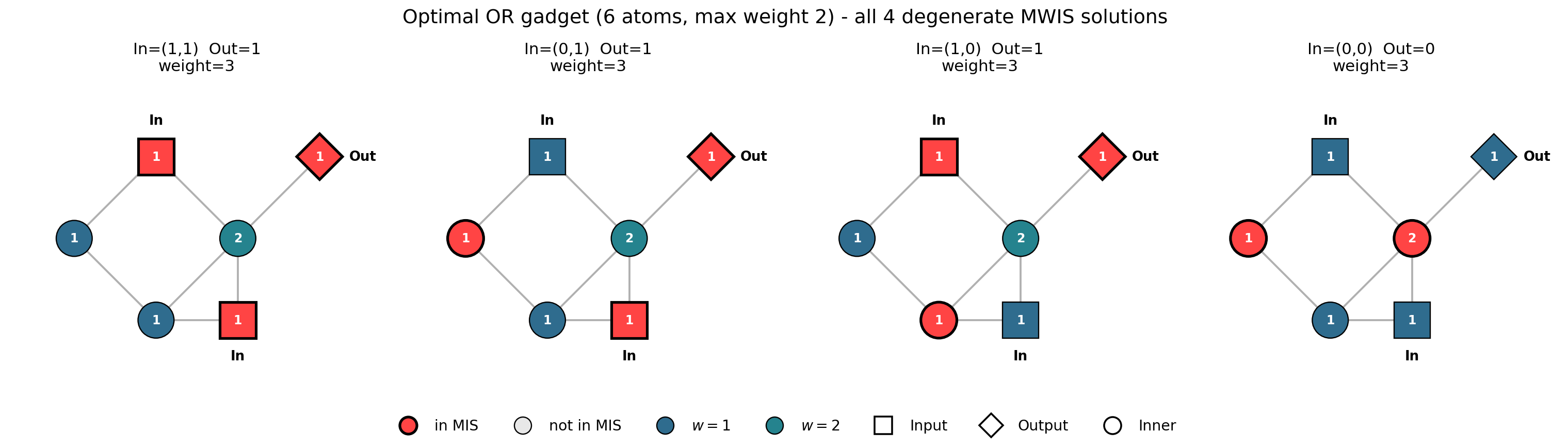}
\caption{The six-atom OR gadget and its four degenerate MWIS solutions, one
per row of the OR truth table, in the same convention as
Fig.~\ref{fig:sm-and}. All four rows $000$, $101$, $011$, $111$ tie at the
maximum weight $C=3$, no atom exceeds weight $2$, and the four optimal
configurations are the only maximal independent sets of the graph. The atom
count, the maximum weight, and the total weight are all certified minimal.}
\label{fig:sm-or}
\end{figure}

\FloatBarrier
\subsection{XOR gadget}
\label{sm:xor-gadget}

XOR is the most constrained primitive of the library: no king-connected
layout of up to eight atoms admits it at all, as certified below. The
gadget adopted here uses nine atoms on a $4\times5$ window: pins
\begin{equation}
x_{1}=(4,2),\qquad x_{2}=(4,5),\qquad y=(1,3),
\end{equation}
and six auxiliary atoms
\begin{equation}
a=(2,3),\quad b=(3,4),\quad c=(1,2),\quad d=(3,2),\quad
e=(2,1),\quad f=(3,3),
\end{equation}
with edges given by king adjacency on these nine sites. The port
king-neighbourhoods among the auxiliary atoms are
$x_{1}:\{d,f\}$, $x_{2}:\{b\}$, and $y:\{a,c\}$, all within the port budget
of two. The solver returns the weights
\begin{equation}
(w_{x_{1}},w_{x_{2}},w_{y},w_{a},w_{b},w_{c},w_{d},w_{e},w_{f})
=(1,1,1,2,2,2,2,2,2).
\end{equation}
The maximum-weight value is $C=6$, attained by exactly four maximal
independent sets [Fig.~\ref{fig:sm-xor}]:
\begin{equation}
000:\{b,c,d\},\qquad
101:\{x_{1},y,b,e\},\qquad
011:\{x_{2},y,e,f\},\qquad
110:\{x_{1},x_{2},a,e\},
\end{equation}
so the ground-state pin patterns are exactly $\{000,101,011,110\}$, i.e.\
$y=x_{1}\oplus x_{2}$. Every other maximal independent set scores at most
$5$, leaving a unit spectral gap (in weight units) between the encoded
truth table and all spurious patterns.

The exhaustive polyplet sweep of Sec.~\ref{sm:search-formulation} certifies
this gadget in atoms and in weight. For sizes $4$--$8$, every layout and
every pin assignment was solved in the most permissive setting; no layout
of at most eight atoms admits XOR at all---at size $8$ alone this amounts
to roughly $3.8$ million infeasible integer programs---so nine atoms is an
unconditional global minimum. All $940982$ nine-cell layouts were then
swept for a total weight of at most $14$ under the library conventions
(port king-degree $\le2$, integer weights $1$--$3$); none exists, so the
total weight $15$ of the assignment above is weight-minimal, and the
maximum weight $2$ is the smallest possible, since an all-unit assignment
would have total weight $9$.

This certified gadget descends from an earlier verified eleven-atom
candidate obeying the same port constraint---pins as above and eight
auxiliary atoms at $(3,4)$, $(1,1)$, $(2,3)$, $(2,5)$, $(1,2)$, $(3,2)$,
$(2,1)$, $(3,3)$, with weights $(1,1,1,3,1,2,1,3,2,3,2)$ and $C=8$---by
deleting two auxiliary atoms, relocating a third, and re-solving all
weights. That predecessor is retained only for reference: all compiled
instances of
Secs.~\ref{sm:cir1-verification}--\ref{sm:multi-verification} are compiled
with the nine-atom gadget above.

\begin{figure}[t]
\centering
\includegraphics[width=0.98\linewidth]{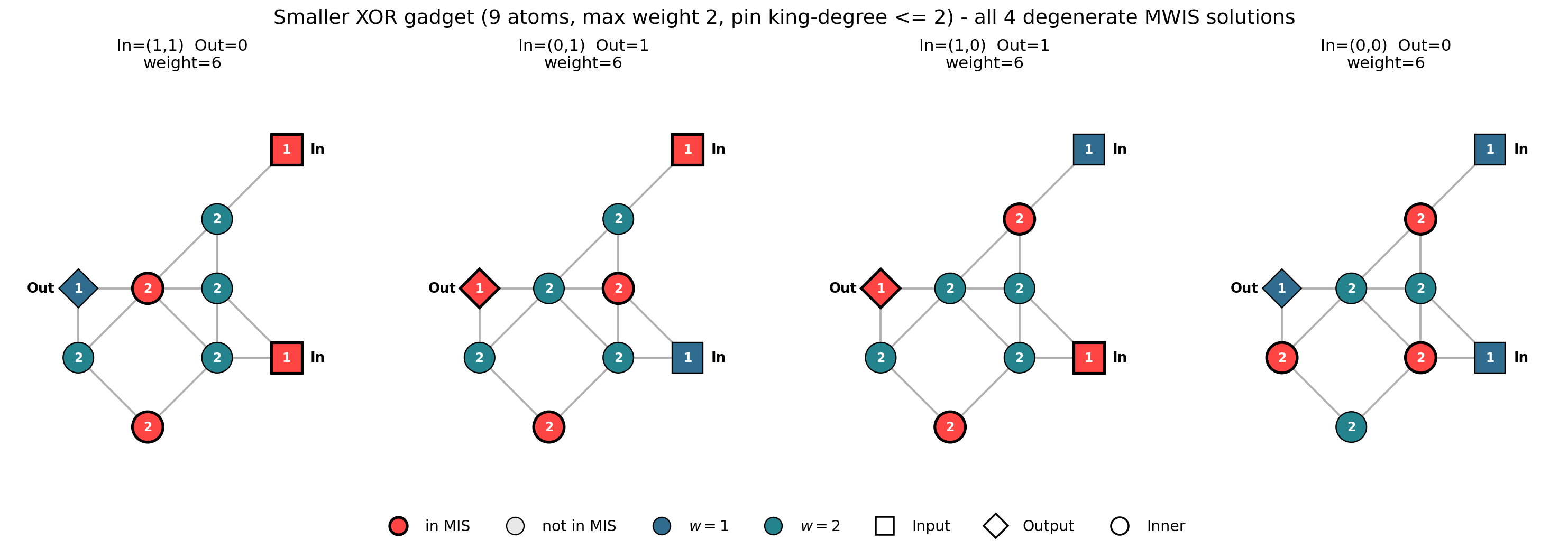}
\caption{The nine-atom XOR gadget and its four degenerate MWIS solutions,
one per row of the XOR truth table, in the same convention as
Fig.~\ref{fig:sm-and}. All four rows $000$, $101$, $011$, $110$ tie at the
maximum weight $C=6$, while the strongest invalid patterns reach only $5$,
so the ground-state pin patterns are exactly the XOR truth table. Every
input and output port is king-adjacent to at most two auxiliary atoms,
respecting the nearest-neighbour budget of the Rydberg implementation, and
the gadget is certified minimal in both atom count and weight.}
\label{fig:sm-xor}
\end{figure}

\FloatBarrier
\subsection{Crossing gadget}
\label{sm:crossing-gadget}

Routing a circuit on a planar king grid inevitably produces wire crossings,
so the compiler requires a four-pin gadget through which two logical signals
pass without interacting. The pins sit clockwise in space---$p_{1}$ top,
$p_{2}$ right, $p_{3}$ bottom, $p_{4}$ left---and the encoded relation is
\begin{equation}
p_{1}=p_{3}
\quad\wedge\quad
p_{2}=p_{4},
\label{eq:sm-crossing-relation}
\end{equation}
i.e.\ the four-bit truth table $T=\{0000,\,1010,\,0101,\,1111\}$ over
$(p_{1},p_{2},p_{3},p_{4})$: the vertical channel copies $p_{1}$ to $p_{3}$,
the horizontal channel copies $p_{2}$ to $p_{4}$, and the two channels are
mutually unconstrained. The same search formulation as for the gate gadgets
applies, with the four-row truth table supplied directly as the constraint.

The gadget uses $11$ atoms on a $5\times5$ king grid: pins
\begin{equation}
p_{1}=(3,5),\qquad p_{2}=(5,3),\qquad p_{3}=(3,1),\qquad p_{4}=(1,3),
\end{equation}
placed at king distance two so that they are pairwise non-adjacent (necessary
for the all-pins row $1111$), and seven auxiliary atoms: four on-axis atoms
$a=(3,4)$, $b=(4,3)$, $c=(3,2)$, $d=(2,3)$ between the centre and each pin,
the central atom $e=(3,3)$, and two adjacent corner atoms $f=(2,4)$,
$g=(4,4)$, with edges given by king adjacency on these eleven sites. The
solver returns the flat weight vector
\begin{equation}
(w_{p_{1}},w_{p_{2}},w_{p_{3}},w_{p_{4}},
 w_{a},w_{b},w_{c},w_{d},w_{e},w_{f},w_{g})
=(1,1,1,1,2,2,2,2,2,2,2),
\end{equation}
i.e.\ every pin has weight $1$ and every auxiliary atom weight $2$. The
maximum-weight value is $C=6$, attained by exactly four maximal
independent sets, one per row of the crossing relation
[Fig.~\ref{fig:sm-crossing}]:
\begin{equation}
0000:\{c,f,g\},\qquad
1010:\{p_{1},p_{3},b,d\},\qquad
0101:\{p_{2},p_{4},a,c\},\qquad
1111:\{p_{1},p_{2},p_{3},p_{4},e\}.
\end{equation}
Every other maximal independent set scores at most $5$, so the four
crossing rows are the unique ground-state pin patterns with a unit gap.
Geometrically, the two channels are realised by complementary on-axis
pairs: when $p_{1},p_{3}$ are occupied (row $1010$) they block the corner
atoms and the on-axis pair $a,c$, and the surviving pair $b,d$ locks
$p_{2},p_{4}$ into the off state; the opposite row $0101$ uses $a,c$
instead of $b,d$. The all-zero row is carried by the triple $\{c,f,g\}$,
and the all-one row picks up the central atom $e$, which is not
king-adjacent to any pin.

The first verified crossing gadget used all nine sites of the central
$3\times3$ block as auxiliary atoms ($13$ atoms, weights
$(2,1,1,1,3,2,2,2,2,1,1,2,2)$, $C=7$); it is retained only for reference,
and the crossing inserted by the router in the compiled multiplier of
Sec.~\ref{sm:multi-verification} is the eleven-atom gadget of this
subsection. Enumerating auxiliary subsets of
increasing size---first within the central block, then over all $21$
non-pin sites of the $5\times5$ grid---shows that no candidate with at most
ten atoms admits any valid weight assignment for this diamond pin layout,
so eleven atoms is minimal, and the accepted seven-atom auxiliary sets are
exactly the four rotations of the shape used here. Relaxing the pin layout
does not help: an exhaustive sweep over all pin quadruples on grids up to
$5\times5$ finds truth-table solutions with as few as seven atoms, but in
all of them the two equal pin pairs sit on the same side---parallel
pass-throughs rather than geometric crossings---and genuine crossings first
appear at ten atoms, where they require a maximum weight of $3$. The
eleven-atom gadget is therefore the lightest genuine crossing and the only
one compatible with the two-value weight alphabet of the library. Its
single deviation from the port budget is the pin $p_{1}$, which is
king-adjacent to the three auxiliary atoms $a$, $f$, $g$.

\begin{figure}[t]
\centering
\includegraphics[width=0.98\linewidth]{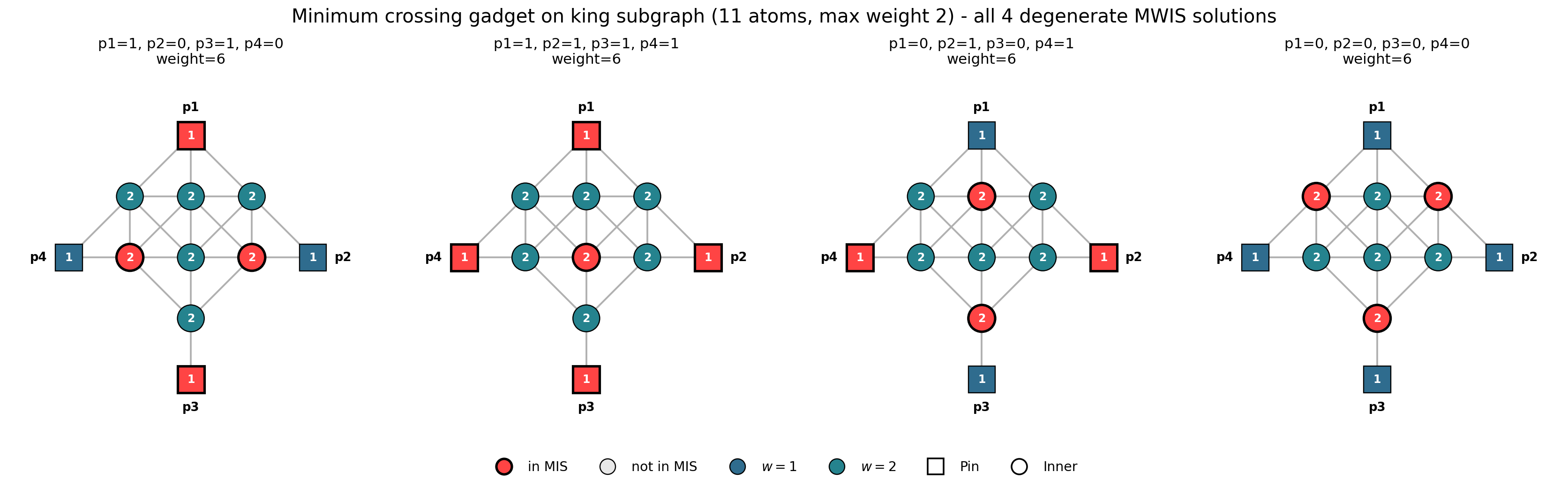}
\caption{The eleven-atom crossing gadget and its four degenerate MWIS
solutions, one per row of the crossing relation
$p_{1}=p_{3}\wedge p_{2}=p_{4}$. Squares are the four pins ($p_{1}$ top,
$p_{2}$ right, $p_{3}$ bottom, $p_{4}$ left, clockwise) and circles are the
seven auxiliary atoms, coloured by weight as in Fig.~\ref{fig:sm-and}. All
four rows $0000$, $1010$, $0101$, $1111$ tie at the maximum weight $C=6$,
so the two logical channels cross without interacting; every pin has weight
$1$ and every auxiliary atom weight $2$.}
\label{fig:sm-crossing}
\end{figure}

\FloatBarrier
\subsection{Fan-out gadget}
\label{sm:fanout-gadget}

The composition rules of Sec.~\ref{subsec:method-solving} require an explicit
fan-out gadget whenever one signal feeds several downstream gates. The fan-out
primitive is an equality constraint on three pins,
\begin{equation}
T=\{111,\,000\},
\qquad\text{i.e.}\qquad
p_{1}=p_{2}=p_{3},
\label{eq:sm-fanout-relation}
\end{equation}
so that one incoming logical value is presented identically on two further
attachment points.

The smallest realisation has four atoms: three pairwise non-adjacent pins
alone can never form a connected graph, and the first success appears at
four atoms as the star $K_{1,3}$---three pins with one central auxiliary
atom king-adjacent to all three, weighted
$(w_{p_{1}},w_{p_{2}},w_{p_{3}},w_{a})=(1,1,1,3)$, whose only maximal
independent sets are $111:\{p_{1},p_{2},p_{3}\}$ and $000:\{a\}$, both of
weight $3$. This star is a predecessor retained only for reference: every
fan-out point of the compiled instances of
Secs.~\ref{sm:cir1-verification}--\ref{sm:multi-verification} uses the
five-atom gadget below. At four
atoms, however, the weight-$3$ centre is unavoidable: the $000$ row must be
represented by the lone auxiliary atom, so
$w_{a}=C=w_{p_{1}}+w_{p_{2}}+w_{p_{3}}\ge3$, and an exhaustive four-atom
sweep confirms that no gadget with maximum weight $\le2$ exists. Splitting
the off-state representative over two auxiliary atoms removes the heavy
atom at the cost of a single extra site.

The five-atom gadget selected for the library has pins
\begin{equation}
p_{1}=(1,1),\qquad p_{2}=(3,1),\qquad p_{3}=(3,3),
\end{equation}
and two auxiliary atoms: $a=(2,2)$, king-adjacent to all three pins, and
$b=(4,2)$, king-adjacent to $p_{2}$ and $p_{3}$, with no further edges. The
solver returns the weights
\begin{equation}
(w_{p_{1}},w_{p_{2}},w_{p_{3}},w_{a},w_{b})=(1,1,1,2,1).
\end{equation}
This graph has exactly three maximal independent sets
[Fig.~\ref{fig:sm-fanout}],
\begin{equation}
111:\{p_{1},p_{2},p_{3}\}\ (\text{weight }3),
\qquad
000:\{a,b\}\ (\text{weight }3),
\qquad
100:\{p_{1},b\}\ (\text{weight }2),
\end{equation}
so the ground-state pin patterns are exactly $\{111,000\}$: the off-state
weight is now carried by the pair $\{a,b\}$ at $2+1=3$ instead of a single
weight-$3$ atom, and the one spurious maximal independent set sits a unit
below the ground-state manifold. Since the four-atom argument above is
independent of any grid bound, five atoms is the certified minimum under
the two-value weight alphabet. Chaining copies of the gadget routes one
signal to arbitrarily many attachment points.

\begin{figure}[t]
\centering
\includegraphics[width=0.65\linewidth]{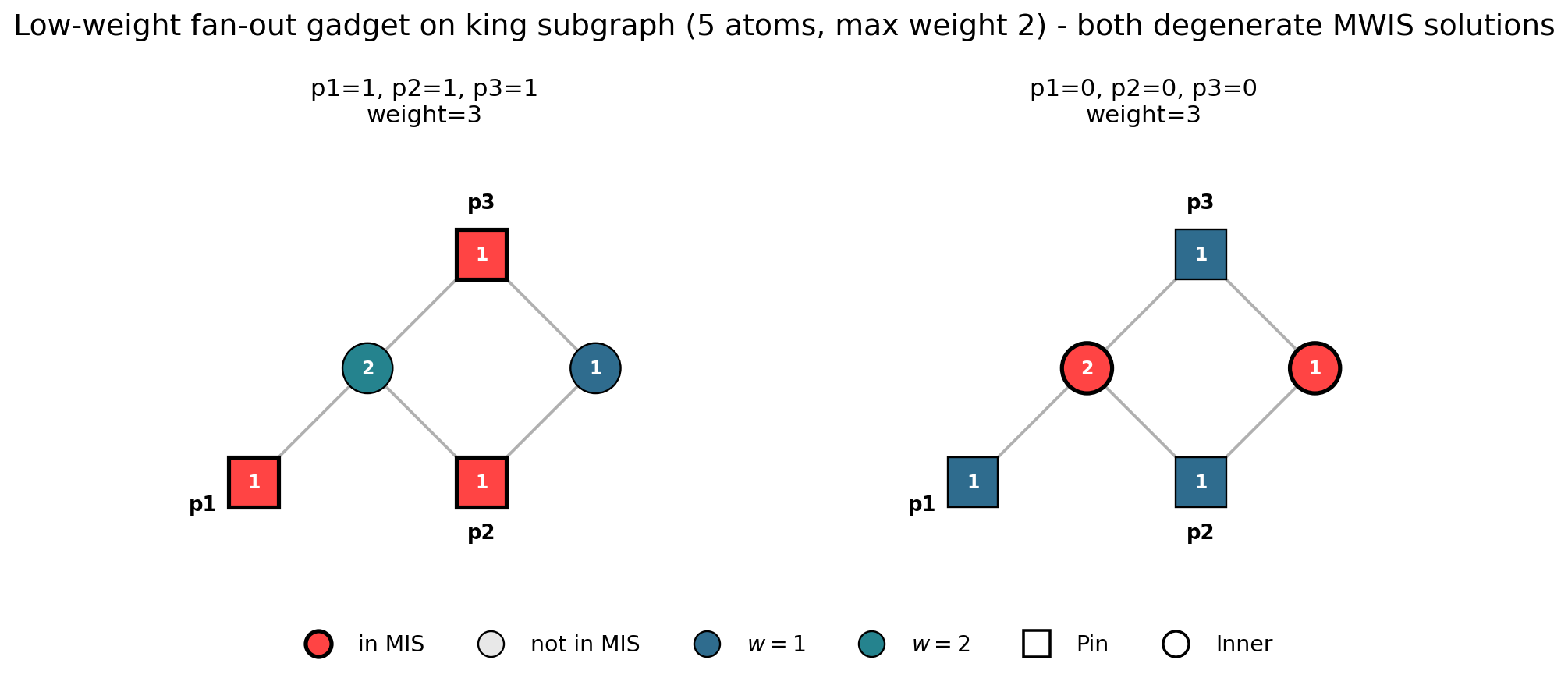}
\caption{The five-atom fan-out gadget and its two degenerate MWIS
solutions, one per row of the equality relation $p_{1}=p_{2}=p_{3}$.
Squares are the three pins and circles are the two auxiliary atoms,
coloured by weight as in Fig.~\ref{fig:sm-and}; the number inside each
vertex is its integer weight. Both rows $111$ and $000$ tie at the maximum
weight $C=3$, while the only other maximal independent set reaches weight
$2$, so the gadget enforces exact equality with every atom weight at most
$2$; a four-atom fan-out gadget provably requires a weight-$3$ atom.}
\label{fig:sm-fanout}
\end{figure}

\FloatBarrier
\subsection{Summary of the gadget library}
\label{sm:gadget-summary}

Table~\ref{tab:sm-gadget-summary} collects the verified parameters of the
six rigid gadgets, together with the variable-length WIRE chain that
carries signals between them. Every atom weight in the library is $1$ or
$2$, with all ports at weight $1$, so the analog implementation requires
only two local-detuning values and ports look identical across the whole
library; altogether the six rigid gadgets use $39$ atoms with total
weight $56$. Port
king-degrees respect the wiring budget of two everywhere except a single
crossing pin of degree three. The ground-state degeneracy of each gadget
equals the number of rows of the encoded truth table, a property verified
by enumeration for every entry of the library; the composition rules of
Sec.~\ref{sm:realizations} require only the weaker equality of optimal
pin-pattern sets. The NOT and OR
gadgets have the stronger property that their optimal rows exhaust the
maximal independent sets of their graphs, and the AND and fan-out gadgets
each admit only a single weight-$2$ spectator set below the ground-state
manifold.

\begin{table}[h]
\centering
\caption{Verified parameters of the weighted KSG gadget library: the four
gate gadgets, the two auxiliary (crossing and fan-out) gadgets, and the
variable-length WIRE chain. $C$ is
the common ground-state weight of the degenerate MWIS manifold; the gap is
the minimum weight separation to the strongest invalid maximal independent
set. The NOT and OR gadgets admit no invalid maximal independent set at
all, so no gap entry applies. The WIRE row is the chain $W_{k}$ of
Eq.~\eqref{eq:wire-path-contract}, a path of $m=2k+1$ interior atoms of
weight $2$ joining two weight-$1$ endpoint pins; the endpoint pins are
identified with the pins they connect, so a chain adds $m$ atoms to the
composite. Its footprint is set by the router rather than by a rigid
template.}
\label{tab:sm-gadget-summary}
\begin{ruledtabular}
\begin{tabular}{lcccccc}
Gadget & Atoms & Weights & $C$ & Gap & Footprint & Ground-state pin patterns \\
\hline
AND & $6$ & $(1,1,1,1,2,2)$ & $3$ & $1$ & $3\times3$ & $000,\,100,\,010,\,111$ \\
NOT & $2$ & $(1,1)$ & $1$ & --- & $2\times1$ & $10,\,01$ \\
OR  & $6$ & $(1,1,1,1,1,2)$ & $3$ & --- & $4\times3$ & $000,\,101,\,011,\,111$ \\
XOR & $9$ & $(1,1,1,2,2,2,2,2,2)$ & $6$ & $1$ & $4\times5$ & $000,\,101,\,011,\,110$ \\
Crossing & $11$ & $(1,1,1,1,2,2,2,2,2,2,2)$ & $6$ & $1$ & $5\times5$ & $0000,\,1010,\,0101,\,1111$ \\
Fan-out & $5$ & $(1,1,1,2,1)$ & $3$ & $1$ & $4\times3$ & $000,\,111$ \\
WIRE & $m+2$ & $(1,2,\ldots,2,1)$ & $m+1$ & $1$ & --- & $00,\,11$ \\
\end{tabular}
\end{ruledtabular}
\end{table}

All compiled instances verified in this Supplemental Material---the
three-gate circuit of Sec.~\ref{sm:cir1-verification}, the full adder of
Sec.~\ref{sm:adder-verification}, and the two-bit multiplier of
Sec.~\ref{sm:multi-verification}---are compiled with the library gadgets of
this section, and all of them inherit its binary weight alphabet
$w\in\{1,2\}$. The predecessor gadgets noted in the preceding
subsections---the $3\times3$ OR variant, the eleven-atom XOR, the
thirteen-atom crossing, and the four-atom $K_{1,3}$ fan-out star---encode the
same truth tables on the same port interfaces and are retained only for
reference.

\FloatBarrier
\section{Compilation pipeline: layout algorithms}
\label{sm:compilation-details}

This section collects the pseudocode of the two-stage placement and
routing compiler described in Sec.~\ref{subsec:method-compilation}.
Algorithm~S2 fixes connectivity and relative topology on the coarse
routing grid, using the router as the evaluation oracle of a two-phase
search. Algorithm~S3 expands each archived abstract layout into a
validated atom level realization and assigns the weights that make the
emitted graph $(G_{C},w)$ satisfy the hypotheses of
Theorem~\ref{thm:realization}.

\begin{center}
\noindent\rule{\columnwidth}{0.4pt}\\[1pt]
\noindent\textbf{Algorithm S2.}~Abstract layout optimization.\\[-2pt]
\noindent\rule{\columnwidth}{0.4pt}
\end{center}
\vspace{-6pt}
\begin{algorithmic}[1]
\footnotesize
\Require circuit $C$, grid $\Gamma$, seeds $\mathcal{S}$
\State normalize fanout, contract external I/O, form two-terminal nets
\Statex \textit{Phase I: collect complete-valid layouts}
\ForAll{$s\in\mathcal{S}$ in parallel}
  \State $\mathcal{L}\gets\Call{Evaluate}{\textproc{InitialPlacement}(C,\Gamma,s)}$
  \For{each annealing step at temperature $T$}
    \State $\mathcal{L}'\gets\Call{Evaluate}{\textproc{ProposeCoordinates}(\mathcal{L})}$
    \State archive $\mathcal{L}'$ if complete-valid and unseen
    \State $\mathcal{L}\gets\mathcal{L}'$ w.p.\
      $P_{\mathrm{acc}}(\mathcal{L}\!\to\!\mathcal{L}';T)$
  \EndFor
\EndFor
\State merge archives, deduplicating by non-crossing anchor coordinates
\Statex \textit{Phase II: deterministic complete-valid descent}
\ForAll{archived candidates $\mathcal{L}$}
  \State alternate route-aware axis compaction and eight-direction
    descent until $\mathbf{q}_{\mathrm{abs}}(\mathcal{L})$ is stationary
\EndFor
\State \Return complete-valid results ranked by
  $(c,L_{\mathrm{tot}},A_{\mathrm{bbox}})$
\end{algorithmic}
\vspace{-4pt}
\noindent\rule{\columnwidth}{0.4pt}
\vspace{4pt}

\begin{center}
\noindent\rule{\columnwidth}{0.4pt}\\[1pt]
\noindent\textbf{Algorithm S3.}~Real layout construction and
optimization.\\[-2pt]
\noindent\rule{\columnwidth}{0.4pt}
\end{center}
\vspace{-6pt}
\begin{algorithmic}[1]
\footnotesize
\Require archived complete-valid abstract layout $\mathcal{L}$
\Statex \textit{Stage 1: construction and validation}
\State expand the grid and stamp all oriented rigid templates
\State reconnect every split net by anchored A$^{\star}$; repair
  even-parity routes by local splice or parity-aware rerouting
\State $\mathcal{R}\gets$ the realization, admitted only if every
  physical validator passes
\Statex \textit{Stage 2: atoms-first greedy descent}
\Repeat
  \State generate wire-cleanup, axis-shrink, and gadget moves
  \State reconnect affected nets; keep the best revalidated improvement
    in $\mathbf{q}_{\mathrm{real}}$
\Until{no improving move remains}
\Statex \textit{Stage 3: weight assignment}
\State assign template and wire weights and \Return $(G_{C},w)$
\end{algorithmic}
\vspace{-4pt}
\noindent\rule{\columnwidth}{0.4pt}
\vspace{4pt}

\section{Ground-state verification of the compiled three-gate circuit}
\label{sm:cir1-verification}

This section provides the complete ground-state verification of the compiled
three-gate Circuit-SAT instance of Fig.~\ref{fig:small-circuit} in the main
text. The circuit [Fig.~\ref{fig:small-circuit}(a)] combines all three
library primitives and computes
\begin{equation}
O_{1} \;=\; (I_{1}\wedge I_{2})\,\oplus\,(I_{2}\vee I_{3}),
\label{eq:sm-cir1-function}
\end{equation}
with the input $I_{2}$ fanned out to both the AND and the OR gate. Its truth
table contains all eight input assignments: $O_{1}=1$ for
$(I_{1},I_{2},I_{3})\in\{001,010,011,101\}$ and $O_{1}=0$ for
$\{000,100,110,111\}$.

The compiled king-subgraph [Fig.~\ref{fig:small-circuit}(b)] contains $30$
atoms and $41$ king edges. It is assembled from the gate library of
Sec.~\ref{sm:gadget-design}---the $6$-atom AND gadget, the $6$-atom OR
gadget, and the $9$-atom XOR gadget---together with the $5$-atom fan-out
gadget of Sec.~\ref{sm:fanout-gadget}, which duplicates $I_{2}$,
and four wire atoms that mediate the routed inter-gadget connections. Upon
composition each gadget keeps its verified template weights; the wire atoms
and the gadget ports engaged in inter-gadget connections are raised from
weight $1$ to $2$ by the MWIS-transparent wire rule of
Sec.~\ref{sm:composition-theory}, while the four primary ports retain weight $1$,
so the assembled graph keeps the binary weight alphabet $w\in\{1,2\}$ of
the library.

The verification is exact: every independent set of the $30$-atom graph is
enumerated by brute force. The maximum weight is $W^{\star}=23$, attained by exactly
eight independent sets---one per input assignment of the circuit
[Fig.~\ref{fig:sm-cir1-mwis} and Table~\ref{tab:sm-cir1-truth}]. Reading the
occupations of the four port atoms ($I_{1},I_{2},I_{3},O_{1}$) in each
optimum reproduces Eq.~\eqref{eq:sm-cir1-function} row by row. The eight
optima contain $13$ or $14$ atoms while tying at the same weight
$W^{\star}=23$, which illustrates that the ground-state degeneracy is organised by
weight rather than by cardinality. Moreover, the exclusion of wrong outputs
is uniform: for every one of the eight \emph{inconsistent} port patterns
(the correct inputs with the output bit flipped), the best independent set
reaches exactly weight $22$, one unit below the ground-state manifold. The
compiled graph therefore realises the dual reading described in the main
text with a uniform unit gap: with free inputs its ground states enumerate
the complete truth table, and with the inputs pinned the energetically
preferred occupation of the output port evaluates the circuit, any incorrect
output costing exactly one weight unit.

\begin{table}[h]
\centering
\caption{Exact ground-state readout of the compiled $30$-atom three-gate
circuit. For each input assignment, the table lists the circuit output
$O_{1}$ of Eq.~\eqref{eq:sm-cir1-function}, the size of the MWIS optimum
whose port atoms read $(I_{1},I_{2},I_{3},O_{1})$, its weight, and the best
weight attainable when the output bit is flipped. All eight truth-table rows
tie at $W^{\star}=23$ and all eight inconsistent patterns stop at $22$.}
\label{tab:sm-cir1-truth}
\begin{ruledtabular}
\begin{tabular}{ccccc}
$(I_{1},I_{2},I_{3})$ & $O_{1}$ & Atoms in optimum & Weight & Output flipped \\
\hline
$(0,0,0)$ & $0$ & $13$ & $23$ & $22$ \\
$(0,0,1)$ & $1$ & $14$ & $23$ & $22$ \\
$(0,1,0)$ & $1$ & $13$ & $23$ & $22$ \\
$(0,1,1)$ & $1$ & $13$ & $23$ & $22$ \\
$(1,0,0)$ & $0$ & $13$ & $23$ & $22$ \\
$(1,0,1)$ & $1$ & $14$ & $23$ & $22$ \\
$(1,1,0)$ & $0$ & $13$ & $23$ & $22$ \\
$(1,1,1)$ & $0$ & $13$ & $23$ & $22$ \\
\end{tabular}
\end{ruledtabular}
\end{table}

\begin{figure}[t]
\centering
\includegraphics[width=0.98\linewidth]{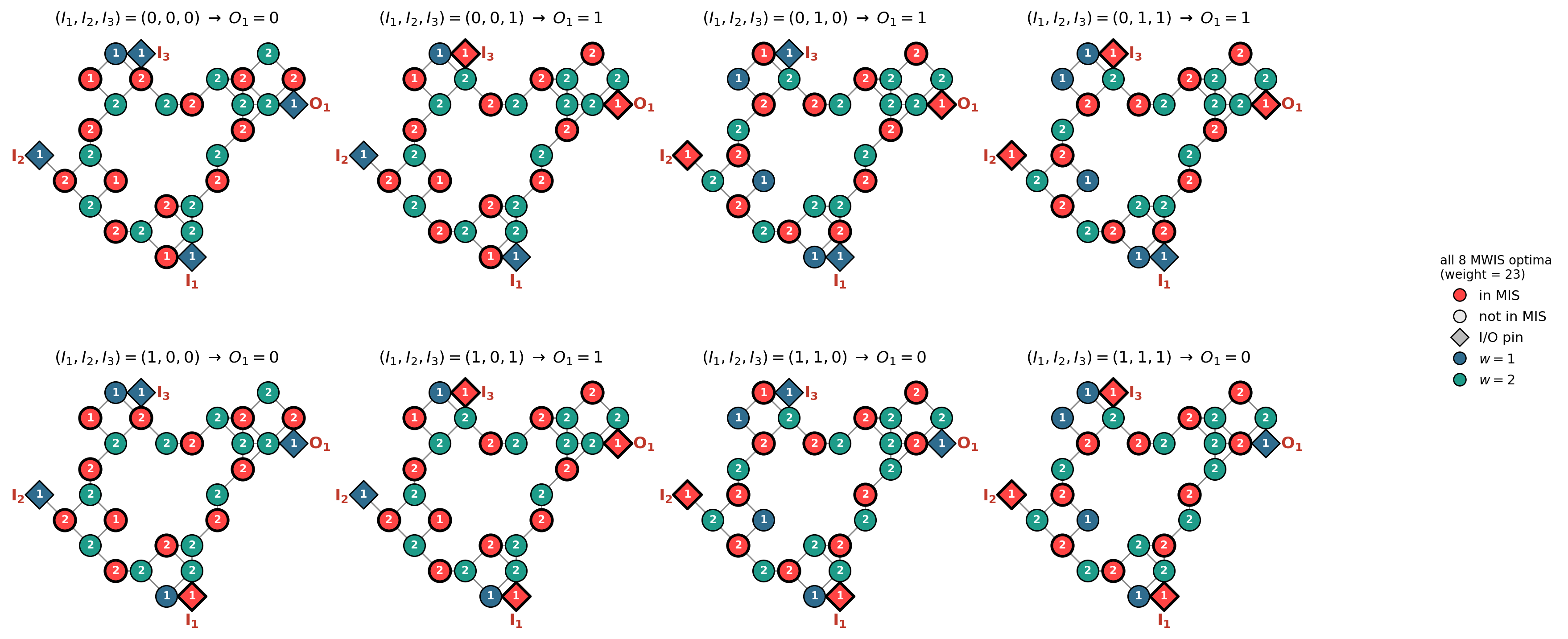}
\caption{All eight degenerate MWIS solutions of the compiled $30$-atom
three-gate circuit of Fig.~\ref{fig:small-circuit}, one per input assignment,
ordered by the input bits $(I_{1},I_{2},I_{3})$. Diamonds mark the I/O pin
atoms, with the four primary ports ($I_{1},I_{2},I_{3},O_{1}$) labelled in
red; circles are internal atoms;
the number inside each atom is its integer weight. Red fill marks atoms in
the selected independent set, and the remaining atoms are coloured by weight
(legend). Every panel attains the maximum weight $W^{\star}=23$, and the port
readout of the eight panels reproduces the complete truth table of
$O_{1}=(I_{1}\wedge I_{2})\oplus(I_{2}\vee I_{3})$.}
\label{fig:sm-cir1-mwis}
\end{figure}

\FloatBarrier
\section{Output-branched certificate for the compiled three-gate circuit}
\label{sm:cir1-branched}

This section applies the output-branching decision procedure of
Sec.~\ref{sec:method} [Eq.~\eqref{eq:branch-criterion}] to the compiled
three-gate circuit of Sec.~\ref{sm:cir1-verification}, providing the
satisfiable counterpart of the unsatisfiable-instance analysis of
Sec.~\ref{sm:unsat-branched}. Asserting $O_{1}=1$ and deleting the output
port atom together with the atoms it blockades removes three atoms of total
weight $5$ from the $30$-atom graph---the weight-$1$ output port and its two
weight-$2$ king-neighbours---leaving a branched graph of $27$ atoms and $34$
king edges.

Exact enumeration of the branched graph yields a maximum weight
$W_{\mathrm{br}}=22$, attained by exactly four degenerate optima of $12$ or
$13$ atoms each [Fig.~\ref{fig:sm-cir1-branched}]. Because the deleted
output port atom carries weight one, the best weight attainable in the full
graph with $O_{1}=1$ is $W_{\mathrm{br}}+1=23=C$, which meets the reference
weight exactly: the branch criterion is satisfied, and the instance is
certified satisfiable from this single optimization, without enumerating the
ground-state manifold of the full graph. This is the complementary outcome
to the unsatisfiable instance of Sec.~\ref{sm:unsat-branched}, where the
branched optimum falls one unit short of the reference.

The branched ground-state manifold moreover carries the complete witness
information: reading the occupations of the three remaining input ports in
the four optima returns exactly the four satisfying assignments
$(I_{1},I_{2},I_{3})\in\{001,010,011,101\}$ of
Eq.~\eqref{eq:sm-cir1-function}---precisely the truth-table rows with
$O_{1}=1$ in Table~\ref{tab:sm-cir1-truth}. Branching thus not only decides
satisfiability but enumerates the full solution set of the satisfiable
instance, at the cost of a single MWIS computation on a strictly smaller
graph.

\begin{figure}[t]
\centering
\includegraphics[width=0.72\linewidth]{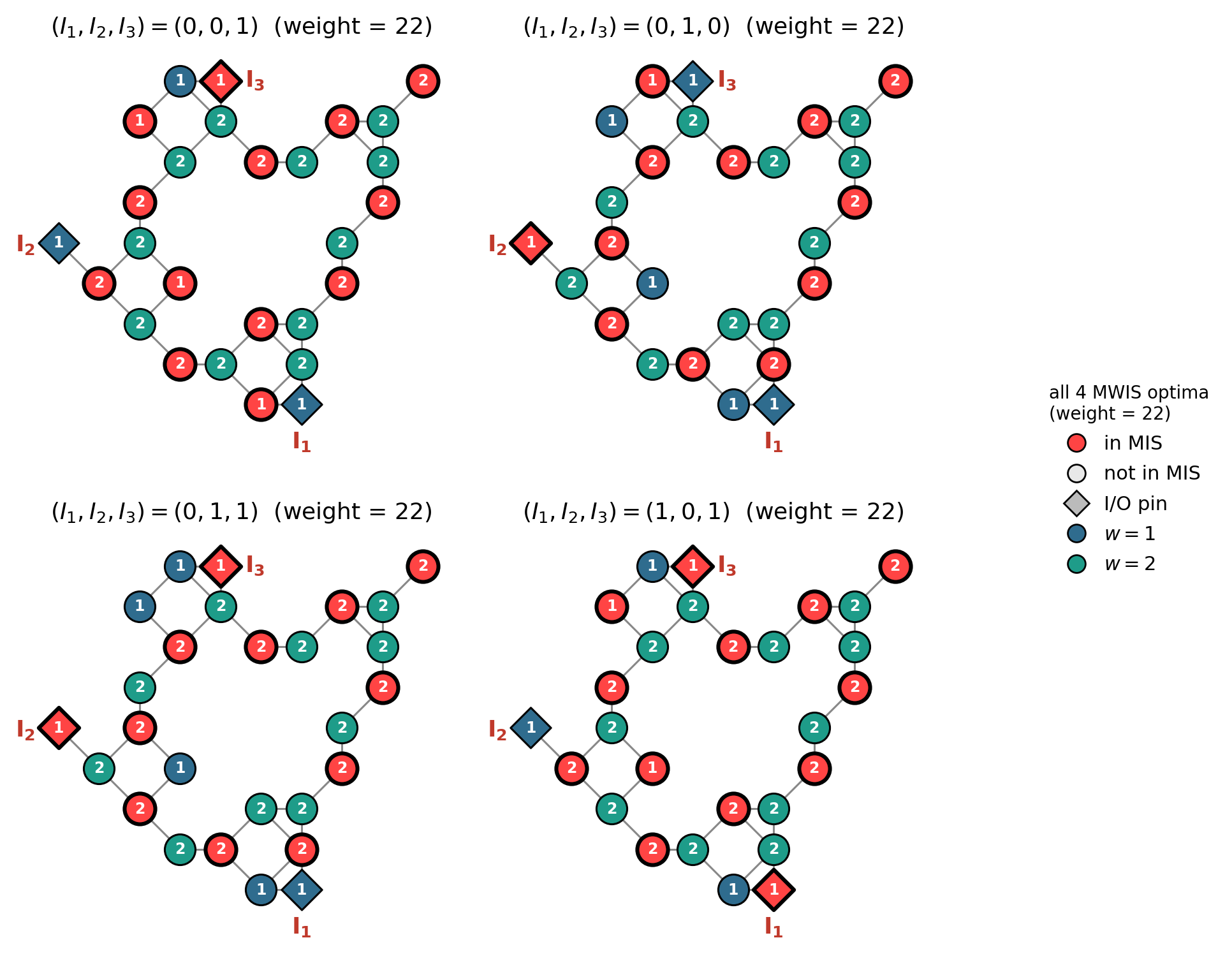}
\caption{The complete set of four degenerate MWIS optima of the
output-branched $27$-atom graph, obtained from the compiled $30$-atom
three-gate circuit of Fig.~\ref{fig:small-circuit} by asserting $O_{1}=1$
and deleting the weight-$1$ output port atom together with the two
weight-$2$ atoms it blockades. Drawing conventions as in
Fig.~\ref{fig:sm-cir1-mwis}; panels are ordered by the input bits
$(I_{1},I_{2},I_{3})$. All four optima tie at the branched maximum weight
$W_{\mathrm{br}}=22$; since the deleted weight-one output port adds exactly
one unit, $W_{\mathrm{br}}+1=23=C$ certifies satisfiability, and the input
readouts of the four panels enumerate exactly the four satisfying
assignments of $O_{1}=(I_{1}\wedge I_{2})\oplus(I_{2}\vee I_{3})$.}
\label{fig:sm-cir1-branched}
\end{figure}

\FloatBarrier
\section{Ground-state verification of the compiled full adder}
\label{sm:adder-verification}

This section extends the ground-state verification of
Sec.~\ref{sm:cir1-verification} to a complete arithmetic block, the one-bit
full adder of Fig.~\ref{fig:arith-blocks}(a,c) in the main
text. The circuit takes the summand bits $I_{1},I_{2}$ and the carry-in
$I_{3}$ and computes the sum and carry-out
\begin{equation}
O_{1} \;=\; I_{1}\oplus I_{2}\oplus I_{3},
\qquad
O_{2} \;=\; (I_{1}\wedge I_{2})\,\vee\,\bigl((I_{1}\oplus I_{2})\wedge
I_{3}\bigr),
\label{eq:sm-adder-function}
\end{equation}
through five gates---two XOR, two AND, and one OR---with four fanned-out
signals: each summand bit feeds both first-stage gates, and the carry-in
and the intermediate sum $I_{1}\oplus I_{2}$ each feed both second-stage
gates.

The automated pipeline compiles this circuit into a weighted king-subgraph of
$85$ atoms and $109$ king edges on a $17\times18$ grid window. The graph
decomposes into the five gate gadgets of the library of
Sec.~\ref{sm:gadget-design}---two $9$-atom XOR, two $6$-atom AND, and one
$6$-atom OR---four $5$-atom fan-out gadgets of
Sec.~\ref{sm:fanout-gadget} realising the fan-outs, and $29$ wire atoms
routing the inter-gadget connections; no crossing gadgets are required. Upon
composition each gadget keeps its verified template weights, and the
MWIS-transparent wire rule of Sec.~\ref{sm:composition-theory} raises the wire atoms
and the inter-gadget ports to weight $2$, so the assembled graph keeps the
binary weight alphabet $w\in\{1,2\}$ of the library.
The five primary ports of the circuit sit on the gadget boundary exactly as
in the three-gate instance: the three inputs terminate on the degree-one
input pins of their fan-out gadgets---$I_{1}$ and $I_{2}$ on the summand
fan-outs, $I_{3}$ on the carry-in fan-out---and the two outputs sit on the
output ports of their gates, $O_{1}$ on the second XOR (sum) and $O_{2}$ on
the OR (carry-out); all five primary ports retain weight $1$ and are marked
as diamonds in Fig.~\ref{fig:sm-adder-mwis}.

This port identification is then confirmed independently by the ground-state
readout. An exact enumeration of all independent sets of the $85$-atom graph
(by the column-transfer dynamic program over the king grid) yields a maximum
weight of $W^{\star}=73$, attained by exactly eight independent sets of $39$--$40$
atoms---one per input assignment. Reading the occupations of the five port
atoms in each optimum reproduces Eq.~\eqref{eq:sm-adder-function} row by row
[Table~\ref{tab:sm-adder-truth} and Fig.~\ref{fig:sm-adder-mwis}]. The
exclusion of wrong outputs is again gapped: resolving the maximum independent
set weight over all $32$ port patterns $(I_{1},I_{2},I_{3},O_{1},O_{2})$, the
eight truth-table rows tie at $73$, and every inconsistent pattern stops at
least one weight unit below. Flipping a single output bit costs exactly one
unit ($72$) for every input row, and flipping both output bits costs two
units ($71$) for the rows $(0,0,0)$ and $(1,1,1)$ and one unit ($72$)
otherwise, so the unit gap of the isolated gate gadgets survives intact in
this $85$-atom five-gate composition.

\begin{table}[h]
\centering
\caption{Exact ground-state readout of the compiled $85$-atom full adder.
For each input assignment, the table lists the sum and carry-out of
Eq.~\eqref{eq:sm-adder-function}, the size of the MWIS optimum whose port
atoms read $(I_{1},I_{2},I_{3},O_{1},O_{2})$, its weight, and the best
weights attainable with one or both output bits flipped. All eight
truth-table rows tie at $W^{\star}=73$ and every inconsistent pattern stops at
least one unit below.}
\label{tab:sm-adder-truth}
\begin{ruledtabular}
\begin{tabular}{cccccc}
$(I_{1},I_{2},I_{3})$ & $O_{1}$ (sum) & $O_{2}$ (carry) & Atoms in optimum &
Weight & One/both outputs flipped \\
\hline
$(0,0,0)$ & $0$ & $0$ & $40$ & $73$ & $72$ / $71$ \\
$(0,0,1)$ & $1$ & $0$ & $40$ & $73$ & $72$ / $72$ \\
$(0,1,0)$ & $1$ & $0$ & $39$ & $73$ & $72$ / $72$ \\
$(0,1,1)$ & $0$ & $1$ & $39$ & $73$ & $72$ / $72$ \\
$(1,0,0)$ & $1$ & $0$ & $39$ & $73$ & $72$ / $72$ \\
$(1,0,1)$ & $0$ & $1$ & $39$ & $73$ & $72$ / $72$ \\
$(1,1,0)$ & $0$ & $1$ & $40$ & $73$ & $72$ / $72$ \\
$(1,1,1)$ & $1$ & $1$ & $40$ & $73$ & $72$ / $71$ \\
\end{tabular}
\end{ruledtabular}
\end{table}

\begin{figure}[t]
\centering
\includegraphics[width=0.98\linewidth]{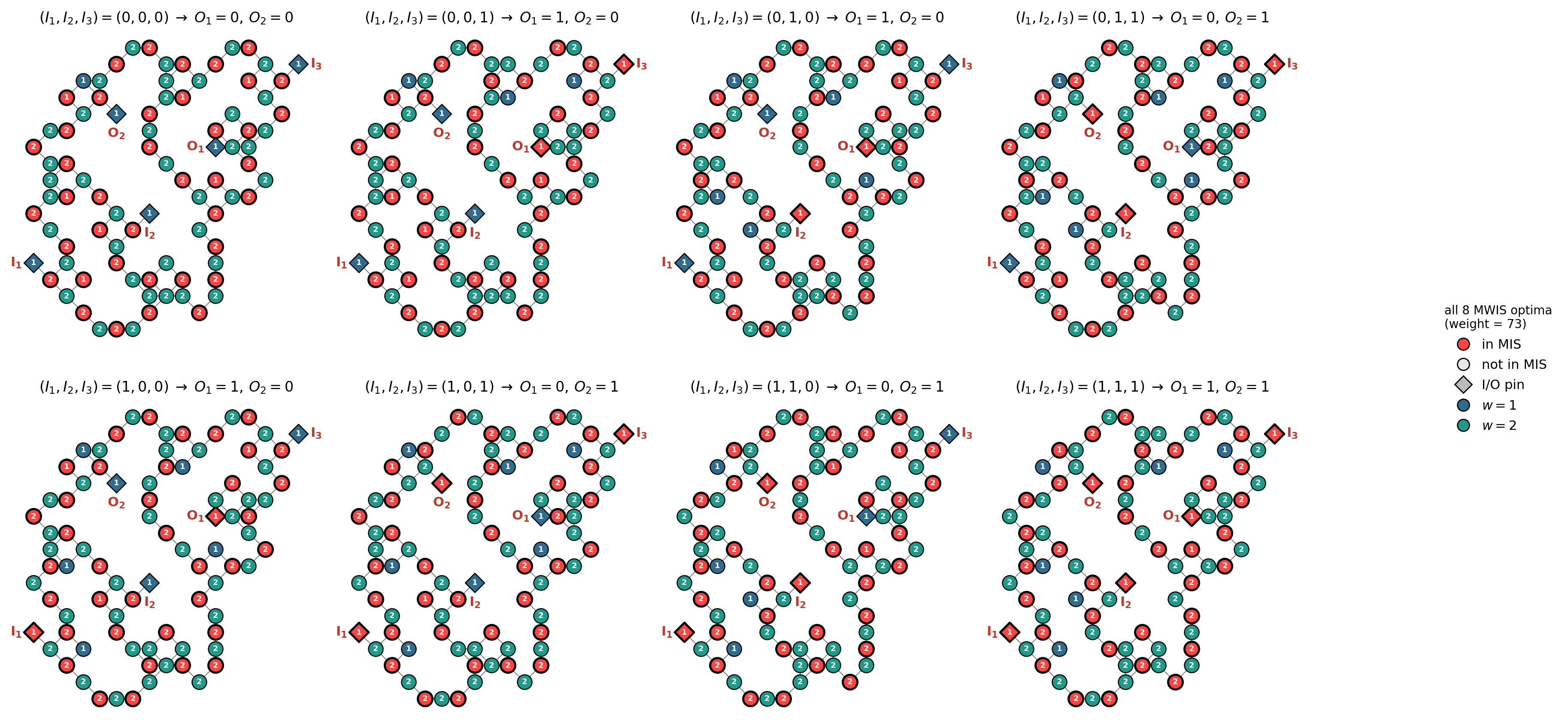}
\caption{All eight degenerate MWIS solutions of the compiled $85$-atom full
adder of Fig.~\ref{fig:arith-blocks}(a,c), one per input assignment, ordered
by the input bits $(I_{1},I_{2},I_{3})$. Diamonds mark the five primary port
atoms ($I_{1},I_{2},I_{3}$ inputs; $O_{1}$ sum; $O_{2}$ carry-out); circles
are internal atoms; the number inside each atom is its integer weight. Red
fill marks atoms in the selected independent set, and the remaining atoms
are coloured by weight (legend). Every panel attains the maximum weight
$W^{\star}=73$, and the port readout of the eight panels reproduces the complete
truth table of the full adder, Eq.~\eqref{eq:sm-adder-function}.}
\label{fig:sm-adder-mwis}
\end{figure}

\FloatBarrier
\section{Ground-state verification of the compiled two-bit multiplier}
\label{sm:multi-verification}

The largest compiled instance of the paper is the two-bit multiplier of
Fig.~\ref{fig:arith-blocks}(b,d) in the main text. The four inputs encode
the two-bit operands $A=(I_{3}\,I_{1})$ and $B=(I_{4}\,I_{2})$, and the four
outputs $O_{1}$ (least significant) through $O_{4}$ (most significant)
encode the product $A\times B$. The circuit uses eight gates: four AND
gates form the partial products
\begin{equation}
O_{1}=I_{1}\wedge I_{2},\qquad
\mathrm{PP}_{10}=I_{3}\wedge I_{2},\qquad
\mathrm{PP}_{01}=I_{1}\wedge I_{4},\qquad
\mathrm{PP}_{11}=I_{3}\wedge I_{4},
\end{equation}
and a two-stage half-adder cascade combines them,
\begin{equation}
O_{2}=\mathrm{PP}_{10}\oplus\mathrm{PP}_{01},\qquad
C_{1}=\mathrm{PP}_{10}\wedge\mathrm{PP}_{01},\qquad
O_{3}=\mathrm{PP}_{11}\oplus C_{1},\qquad
O_{4}=\mathrm{PP}_{11}\wedge C_{1}.
\label{eq:sm-multi-function}
\end{equation}
The automated pipeline compiles this circuit into a weighted king-subgraph
of $165$ atoms and $218$ king edges on a $23\times28$ grid window. The graph
decomposes into the eight gate gadgets of the library of
Sec.~\ref{sm:gadget-design}---two $9$-atom XOR and six $6$-atom AND---eight
$5$-atom fan-out gadgets of Sec.~\ref{sm:fanout-gadget}, a single
$11$-atom crossing gadget of Sec.~\ref{sm:crossing-gadget} inserted by the
router where two independent signals must cross, and $60$ wire atoms routing
the $22$ inter-gadget connections. Upon composition each gadget keeps its
verified template weights, and the MWIS-transparent wire rule of
Sec.~\ref{sm:composition-theory} raises the wire atoms and the inter-gadget ports
to weight $2$, so the assembled graph keeps the binary weight alphabet
$w\in\{1,2\}$ of the library.
The eight primary ports of the circuit sit on the gadget boundary exactly
as in the previous instances: the four inputs terminate on the degree-one
input pins of the four operand fan-out gadgets ($I_{1}$ and $I_{3}$ on the
fan-outs of the first operand, $I_{2}$ and $I_{4}$ on those of the second),
and the four product bits sit on the output ports of their gates ($O_{1}$
on the first partial-product AND, $O_{2}$ and $O_{3}$ on the two XOR gates,
and $O_{4}$ on the final AND); all eight primary ports retain weight $1$
and are marked as diamonds in Fig.~\ref{fig:sm-multi-mwis}.

This port identification is then confirmed independently by the
ground-state readout. The exact enumeration (by the column-transfer dynamic
program over the king grid) yields a maximum
weight of $W^{\star}=142$, attained by exactly sixteen independent sets of
$74$--$78$ atoms---one per input assignment.

The readout of the sixteen optima reproduces the
complete two-bit multiplication table
[Table~\ref{tab:sm-multi-truth} and Fig.~\ref{fig:sm-multi-mwis}],
including the row $A=B=3$ with product $9=1001_{2}$, which exercises the
full carry cascade. The exclusion of wrong outputs is again gapped:
resolving the maximum independent-set weight over all $2^{8}=256$ port
patterns $(I_{1},\dots,I_{4},O_{1},\dots,O_{4})$, the sixteen correct rows
tie at $142$ and every one of the $240$ patterns with at least one wrong
output bit stops at least one weight unit below, peaking between $138$ and
$141$. Misreading any subset of the product bits therefore always costs at
least one weight unit, so the gap of the gadget library survives in the
largest composition of the paper.

\begin{table}[h]
\centering
\caption{Exact ground-state readout of the compiled $165$-atom two-bit
multiplier. For each input assignment, the table lists the operands
$A=(I_{3}\,I_{1})$ and $B=(I_{4}\,I_{2})$, the product bits
$(O_{1},O_{2},O_{3},O_{4})$ read from the port atoms (least significant
first), the decimal product, and the size of the MWIS optimum. All sixteen
rows tie at $W^{\star}=142$, and every pattern with any wrong output bit peaks at
least one weight unit below.}
\label{tab:sm-multi-truth}
\begin{ruledtabular}
\begin{tabular}{ccccccc}
$(I_{1},I_{2},I_{3},I_{4})$ & $A$ & $B$ & $(O_{1},O_{2},O_{3},O_{4})$ &
$A\times B$ & Atoms in optimum & Weight \\
\hline
$(0,0,0,0)$ & $0$ & $0$ & $(0,0,0,0)$ & $0$ & $78$ & $142$ \\
$(0,0,0,1)$ & $0$ & $2$ & $(0,0,0,0)$ & $0$ & $77$ & $142$ \\
$(0,0,1,0)$ & $2$ & $0$ & $(0,0,0,0)$ & $0$ & $77$ & $142$ \\
$(0,0,1,1)$ & $2$ & $2$ & $(0,0,1,0)$ & $4$ & $76$ & $142$ \\
$(0,1,0,0)$ & $0$ & $1$ & $(0,0,0,0)$ & $0$ & $77$ & $142$ \\
$(0,1,0,1)$ & $0$ & $3$ & $(0,0,0,0)$ & $0$ & $76$ & $142$ \\
$(0,1,1,0)$ & $2$ & $1$ & $(0,1,0,0)$ & $2$ & $76$ & $142$ \\
$(0,1,1,1)$ & $2$ & $3$ & $(0,1,1,0)$ & $6$ & $75$ & $142$ \\
$(1,0,0,0)$ & $1$ & $0$ & $(0,0,0,0)$ & $0$ & $77$ & $142$ \\
$(1,0,0,1)$ & $1$ & $2$ & $(0,1,0,0)$ & $2$ & $76$ & $142$ \\
$(1,0,1,0)$ & $3$ & $0$ & $(0,0,0,0)$ & $0$ & $76$ & $142$ \\
$(1,0,1,1)$ & $3$ & $2$ & $(0,1,1,0)$ & $6$ & $75$ & $142$ \\
$(1,1,0,0)$ & $1$ & $1$ & $(1,0,0,0)$ & $1$ & $77$ & $142$ \\
$(1,1,0,1)$ & $1$ & $3$ & $(1,1,0,0)$ & $3$ & $76$ & $142$ \\
$(1,1,1,0)$ & $3$ & $1$ & $(1,1,0,0)$ & $3$ & $76$ & $142$ \\
$(1,1,1,1)$ & $3$ & $3$ & $(1,0,0,1)$ & $9$ & $74$ & $142$ \\
\end{tabular}
\end{ruledtabular}
\end{table}

\begin{figure}[t]
\centering
\includegraphics[width=0.92\linewidth]{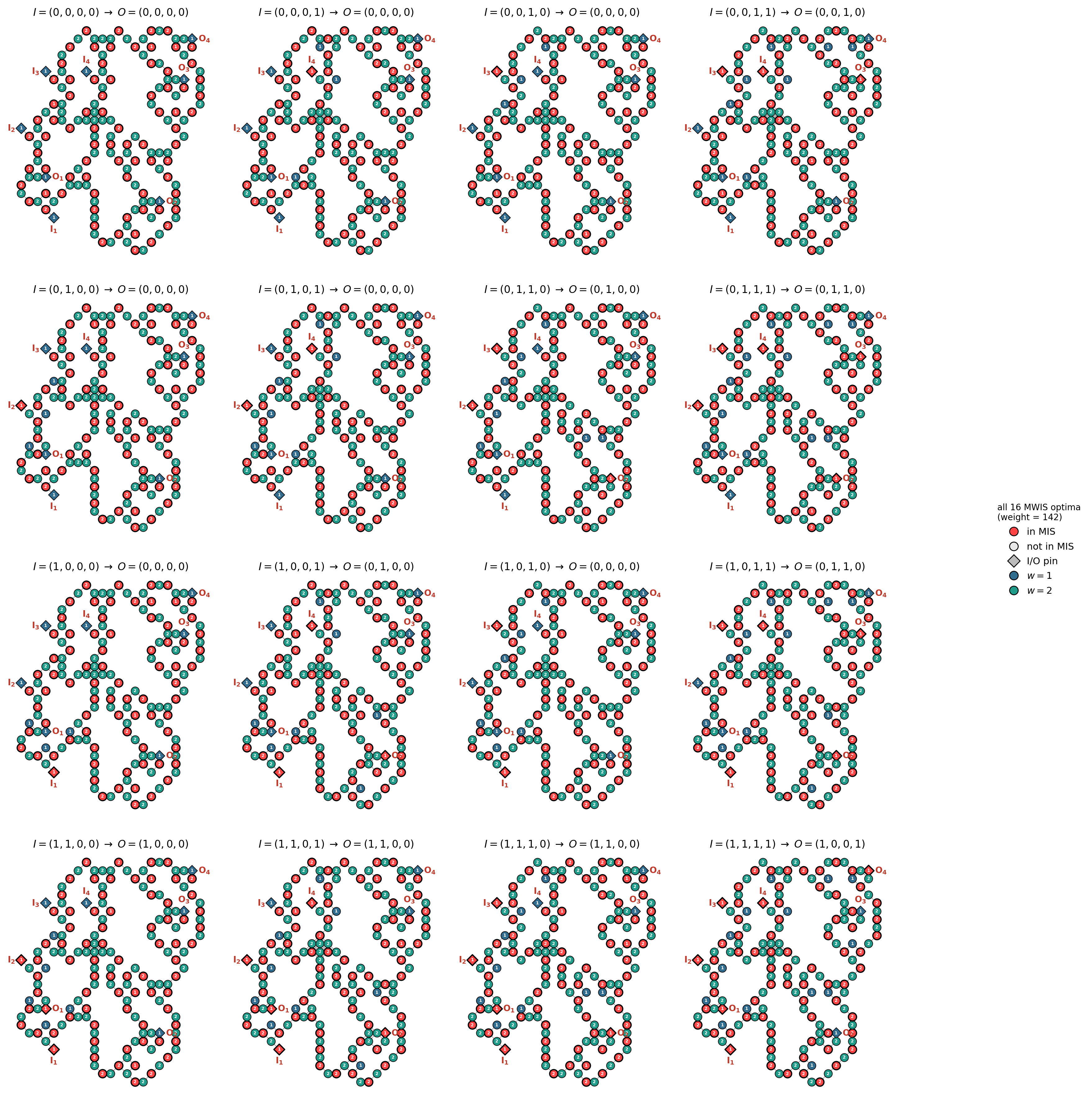}
\caption{All sixteen degenerate MWIS solutions of the compiled $165$-atom
two-bit multiplier of Fig.~\ref{fig:arith-blocks}(b,d), one per input
assignment, ordered by the input bits $(I_{1},I_{2},I_{3},I_{4})$. Diamonds
mark the eight primary port atoms (inputs $I_{1}$--$I_{4}$; product bits
$O_{1}$--$O_{4}$, least significant first); circles are internal atoms; the
number inside each atom is its integer weight. Red fill marks atoms in the
selected independent set, and the remaining atoms are coloured by weight
(legend). Every panel attains the maximum weight $W^{\star}=142$, and the port
readout of the sixteen panels reproduces the complete two-bit
multiplication table, Eq.~\eqref{eq:sm-multi-function}.}
\label{fig:sm-multi-mwis}
\end{figure}

\FloatBarrier
\section{Complete enumeration of optima for the unsatisfiable instance}
\label{sm:unsat-enumeration}

This section reports the full ground-state manifolds that underlie the
unsatisfiability analysis of Sec.~\ref{subsec:unsat} and
Fig.~\ref{fig:unsat} in the main text, for the minimal unsatisfiable
instance $O_{1}=I_{1}\oplus I_{1}$. Figure~\ref{fig:unsat}(c) shows only a
single representative optimum of the unconstrained and of the
output-branched graph; here we enumerate \emph{every} degenerate
maximum-weight independent set of both, so that the main-text claims---that
every optimum reads $O_{1}=0$, and that the branched optimum reaches only
$W_{\mathrm{br}}=15$---can be checked configuration by configuration.

\subsection{Unconstrained graph}
\label{sm:unsat-full}

Exact enumeration of the $20$-atom compiled king-subgraph
[Fig.~\ref{fig:unsat}(b)] yields a maximum weight $W^{\star}=17$, attained by
exactly two degenerate maximum-weight independent sets, one per input value
$I_{1}=0,1$, each occupying nine atoms [Fig.~\ref{fig:sm-unsat-full}]. Both
optima are gate-consistent: the fanout and XOR gadgets each sit in a valid
truth-table row, and the input and output pin chains each carry a single
well-defined Boolean value. In both optima the output port reads $O_{1}=0$,
the correct value of $I_{1}\oplus I_{1}$, and neither places the output port
in the independent set. The absence of any weight-$17$ optimum with
$O_{1}=1$ is precisely the graph-level signature of unsatisfiability
discussed in the main text.

\begin{figure}[t]
\centering
\includegraphics[width=0.98\linewidth]{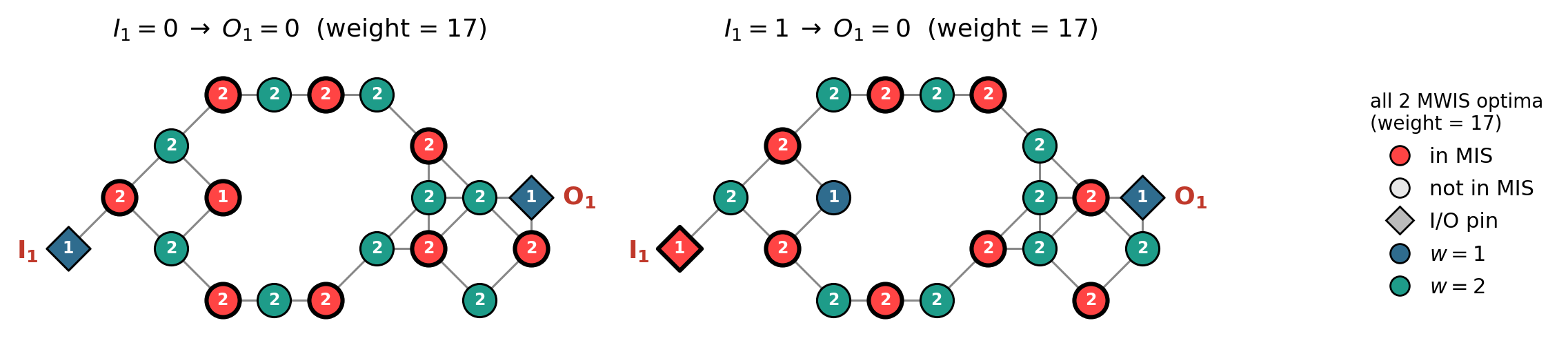}
\caption{The complete set of two degenerate MWIS optima of the unconstrained
$20$-atom king-subgraph compiled from the unsatisfiable instance
$O_{1}=I_{1}\oplus I_{1}$ [Fig.~\ref{fig:unsat}(b)], one per input value
$I_{1}=0,1$. The number inside each
atom is its integer weight; red fill marks atoms in the selected independent
set, and the remaining atoms are coloured by weight (legend). Both optima tie
at the maximum weight $W^{\star}=17$, are gate-consistent, and read the
output port $O_{1}=0$; no
weight-$17$ configuration places the output port at $1$.}
\label{fig:sm-unsat-full}
\end{figure}

\FloatBarrier
\subsection{Output-branched graph}
\label{sm:unsat-branched}

Asserting $O_{1}=1$ and deleting the output port atom together with the atoms
it blockades leaves the $17$-atom branched graph [Fig.~\ref{fig:unsat}(c),
bottom]
(three atoms of total weight $5$ removed: the weight-$1$ output port and its
two weight-$2$ king-neighbours). Its maximum weight is
$W_{\mathrm{br}}=15$, attained by exactly fourteen degenerate optima of
eight atoms each, seven for each input value
[Fig.~\ref{fig:sm-unsat-branched}]---the much larger branched degeneracy
reflects the freedom left in the gadget interiors once the output region is
removed, and is irrelevant to the certificate, which uses only the branched
maximum weight. Because the deleted
output port atom carries weight one, an attainable $O_{1}=1$ would require
$W_{\mathrm{br}}+1=W^{\star}$; here $W_{\mathrm{br}}+1=16<17=W^{\star}$, so all
fourteen branched optima fall one unit short of the reference, and the
instance is
certified unsatisfiable from this single optimization---without ever
enumerating the full optimal set of the unconstrained graph.

\begin{figure}[t]
\centering
\includegraphics[width=0.98\linewidth]{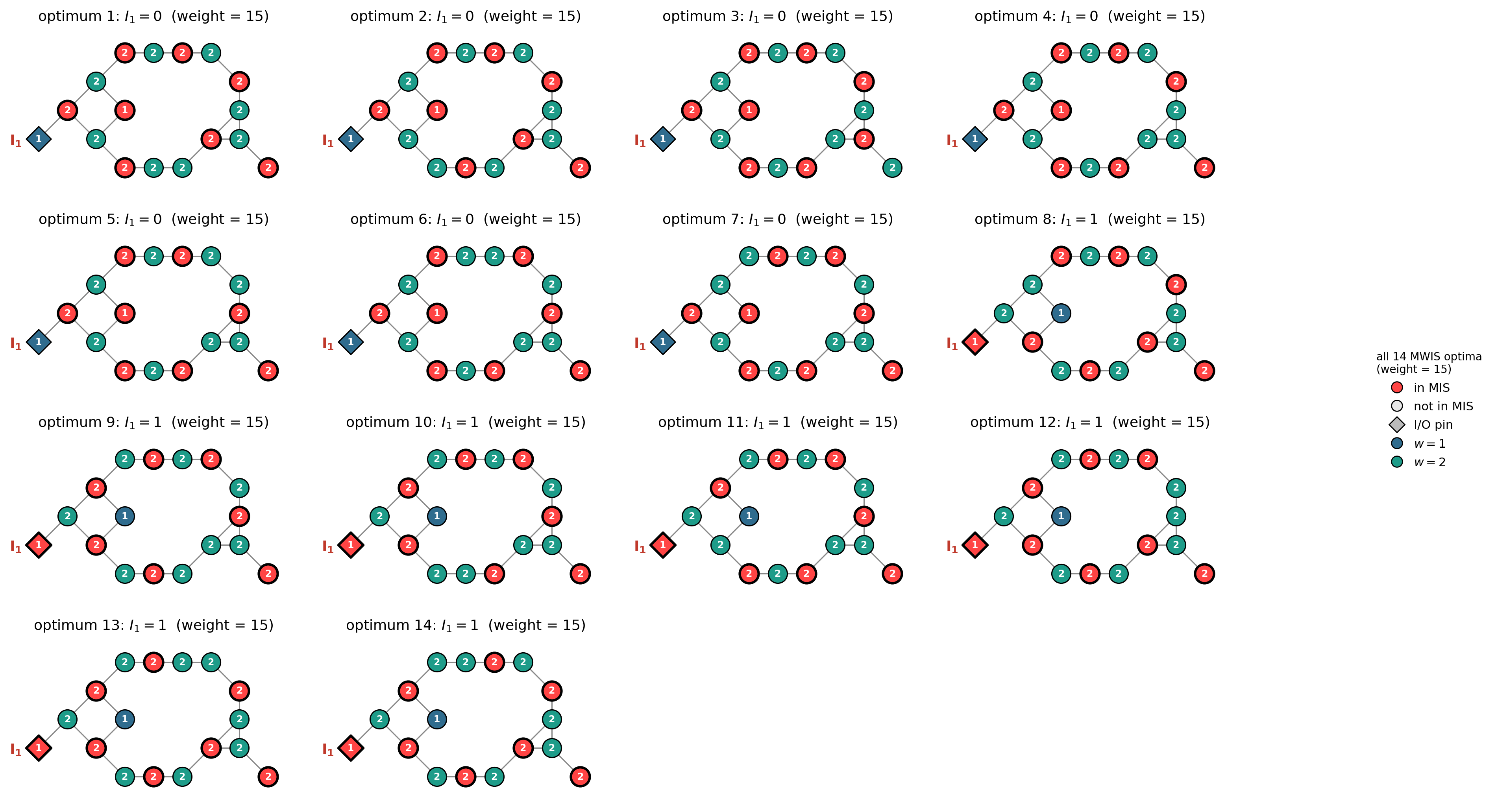}
\caption{The complete set of fourteen degenerate MWIS optima of the
output-branched $17$-atom graph [Fig.~\ref{fig:unsat}(c), bottom], obtained
from the
full graph by asserting $O_{1}=1$ and deleting the weight-$1$ output port
atom together
with the two weight-$2$ atoms it blockades. Drawing conventions as in
Fig.~\ref{fig:sm-unsat-full}; panels are ordered by the input bit, with
seven optima for each of $I_{1}=0$ and $I_{1}=1$. All fourteen optima tie at
the branched maximum
weight $W_{\mathrm{br}}=15$ and occupy eight atoms each; since the
deleted weight-one output port would add at most one unit, the best weight
attainable with $O_{1}=1$ is $16<17=W^{\star}$, the one-unit deficit that
certifies unsatisfiability.}
\label{fig:sm-unsat-branched}
\end{figure}

\FloatBarrier

\end{document}